\pdfoutput=1
\documentclass[11pt,a4paper]{article}

\usepackage[T1]{fontenc}
\usepackage[utf8]{inputenc}
\usepackage{lmodern}
\usepackage[USenglish]{babel}
\usepackage[margin=1in]{geometry}
\usepackage{amsmath}
\usepackage{amsthm}
\usepackage{graphicx}
\usepackage{xcolor}
\usepackage{authblk}
\usepackage[hidelinks]{hyperref}

\newcommand{\Z}{\mathbb{Z}}

\usepackage{xspace}
\usepackage{cleveref}
\usepackage{todonotes}
\usepackage{thmtools}
\usepackage{booktabs}

\declaretheorem[name=Theorem,style=plain]{theorem}
\declaretheorem[name=Lemma,style=plain,sibling=theorem]{lemma}
\declaretheorem[name=Corollary,style=plain,sibling=theorem]{corollary}

\declaretheorem[name=Example,style=definition,sibling=theorem]{example}
\declaretheorem[name=Remark,style=remark,sibling=theorem]{remark}

\usepackage{tabularx}
\usepackage{amssymb}
\usepackage{tikz}
\usetikzlibrary{arrows.meta}

\newcommand\parent{\textsf{parent}}

\newcommand\view[2]{\textsf{view}^{#1}(#2)}
\newcommand\alg[1]{A(#1)}
\newcommand\hier[1]{#1} 
\newcommand\modelo{\mathsf{M}}

\newcommand\frag[2]{\textsf{Fragile}_{#1}(#2)}
\newcommand\size{\textsc{size}}

\newcommand\leader{\textsc{leader}}

\newcommand\spann{\textsc{SetTree}}

\newcommand\alos{\textsc{alos}}
\newcommand\amos{\textsc{amos}}
\newcommand\xor{\textsf{XOR}}

\newcommand\seth{\textsf{SET}\xspace}
\newcommand\multiseth{\textsf{MULTISET}\xspace}

\newcommand\portsh{\textsf{PORTS}\xspace}

\newcommand\lcount[1]{\textsf{Count}(#1)}
\newcommand{\cardmod}[3]{\textsf{Card}^{#1}_{#2}(#3)}
\newcommand\quantif{\mathcal{Q}}
\newcommand\lbl[1]{\mathbb{B}_{}^V}
\newcommand{\dist}{\mathsf{dist}}
\newcommand\portspann{\textsc{PortTree}}

\newcommand{\Accept}{\mathsf{Accept}}
\newcommand{\Type}{\mathsf{Type}}
\newcommand{\SetView}{\mathsf{SetView}}

\newcommand{\Part}{\mathsf{Part}}
\newcommand{\Cert}{\mathsf{Cert}}
\newcommand{\sethr}[2]{\mathsf{SET}^{\,#1}_{#2}}
\title{On the Power of Alternating Local Certification with Constant-Size Certificates} 

\author[1]{Pedro Montealegre}
\author[1]{Mart\'{\i}n R\'{\i}os-Wilson}
\author[2]{Jukka Suomela}
\author[3]{Guillaume Theyssier}
\affil[1]{Facultad de Ingenier\'{\i}a y Ciencias, Universidad Adolfo Ib\'a\~nez}
\affil[2]{Aalto University}
\affil[3]{I2M, CNRS, Marseille, France}
\date{}

\begin{document}

\maketitle

\begin{abstract}
We study a distributed decision model based on the interaction
between a prover, a disprover, and a distributed verifier.
During a constant number of alternating rounds, the prover
and the disprover assign certificates to the nodes, responding
to the preceding assignments. The verifier is a local
distributed algorithm in which each node decides whether to
accept from its own input and the certificates observed in its
neighborhood. The instance is accepted if all nodes accept.

Alternation hierarchies for distributed decision have been
studied under different restrictions, with previous works
exploring, for example, the impact of local computation
(Reiter, PODC 2024) or logarithmic-size certificates
(Feuilloley et al., ICALP 2016).

In this paper, we study a restricted setting where certificates
have constant size and the network is 
anonymous. Moreover, we consider three models limiting the
information available to the verifier. In the first model,
called the \seth hierarchy, nodes observe only the \emph{set}
of certificate sequences assigned to their neighbors. In the
second, strictly stronger model, called the \multiseth
hierarchy, nodes observe the corresponding \emph{multiset}.
In the third model, called the \portsh hierarchy, nodes can
additionally distinguish neighboring certificate sequences
by local port numbers. We prove that \portsh is strictly
stronger than \multiseth, and that \multiseth is strictly stronger than \seth.

Despite these severe restrictions, the resulting hierarchies remain surprisingly expressive. In particular, leader detection belongs to the third level of the \seth hierarchy, together with several natural graph properties previously known to require logarithmic-size certificates or unique identifiers.

Our main results are the following.
\begin{itemize}
    \item We characterize the \seth hierarchy: over finite
    input alphabets, a distributed language belongs to
    \seth if and only if it is definable in monadic
    second-order logic \(\mathrm{MSO}_1\).

    \item Every language definable in counting monadic
    second-order logic \(\mathrm{CMSO}_1\) belongs to
    the \portsh hierarchy. We also exhibit a language
    in \multiseth that is not \(\mathrm{CMSO}_1\)-definable,
    while the inclusion of \(\mathrm{CMSO}_1\) in
    \multiseth remains open.

    \item Each of the \seth, \multiseth, and \portsh
    hierarchies is strict at every level.

    \item We develop lower-bound techniques showing that
    neither Hamiltonicity nor the existence of a perfect
    matching belongs to \portsh. Since both properties
    are \(\mathrm{MSO}_2\)-definable, the containment of
    \(\mathrm{MSO}_1\) in \portsh does not extend to
    \(\mathrm{MSO}_2\).

    \item For unbounded input alphabets, verifying that
    every node is given the correct network size lies
    outside the \portsh hierarchy.
\end{itemize}

These results motivate further study of weak models of
alternating local certification, with open questions
about how local information and alternation determine
their expressive power.

\medskip

\noindent \textbf{Use of AI tools.} The authors used ChatGPT (OpenAI) to assist with the
writing and editing of this paper and in finding
a proof of Theorem~7. The authors take full
responsibility for the correctness of all results
and proofs.

\end{abstract}

\clearpage
\setcounter{tocdepth}{2} 
\tableofcontents
\thispagestyle{empty}

\clearpage
\setcounter{page}{1}

\section{Introduction}
\label{sec:intro}

In distributed local certification, the nodes of a network aim to collectively verify whether the network satisfies a given property, \emph{with the help of a prover}. The prover is an untrusted oracle that tries to convince the nodes
that the network satisfies the property.

Several models of local certification have been studied. The first one was introduced by Korman, Kutten, and Peleg~\cite{korman2005proof} under the name \emph{proof-labeling schemes} (PLS). A PLS protocol proceeds in two stages. In the first stage, the prover assigns a certificate to each node. In the second stage, called the verification stage, every node inspects the certificates in
its local neighborhood and decides whether to accept or reject. If the instance is legal then there must exist a certificate assignment that makes all nodes accept. Conversely, if the instance is not legal, then
for every certificate assignment, at least one node must reject. The efficiency of a protocol is measured by the size of the certificates.

A canonical example is the verification of bipartiteness. This property can be verified by assigning to each node a single bit, interpreted as its color in a purported proper 2-coloring. A node accepts if and only if all its neighbors receive a color different from its own. Thus, all nodes accept if and only if the input graph is bipartite. This should be contrasted with the situation without certificates: no constant-round local algorithm can decide bipartiteness, since detecting an odd cycle may require exploring arbitrarily large portions of the graph; equivalently, long even and odd cycles are locally
indistinguishable for nodes whose view is restricted to a constant radius
(see, e.g.,~\cite{feuilloley2021introduction}).

Many graph properties admit efficient verification in the PLS model with logarithmic-sized certificates. Examples include planarity~\cite{feuilloley2021compact}, bounded
genus~\cite{feuilloley2023local}, graph classes defined by intersections of geometric objects~\cite{jauregui2023local}, \(H\)-minor-freeness~\cite{bousquet2024local}, model checking on bounded treewidth~\cite{cook2025tight}, etc. In all
these cases, certificates make it possible to verify properties that cannot be recognized by constant-round local algorithms without the help of a prover. However, for most known local certification schemes, constant-size certificates are not sufficient; logarithmic or polylogarithmic certificates are typically used, and are often necessary.

Let us illustrate why finding properties that can be verified with
constant-size certificates is difficult. Consider the property ''\emph{there exists a node of degree 3}''. Although this property is global, it can be verified in the PLS model by certifying a spanning tree rooted at a node of degree \(3\). The tree is encoded by giving each node, except the root, the identifier of its parent and its distance to the root, assuming that the input graph is connected. Each node
checks that its parent's distance is exactly one less than its own, while the root checks that its distance is \(0\) and that it has degree \(3\). Hence, by using a spanning tree, one can easily certify non-local properties. However, Korman, Kutten, and Peleg~\cite{korman2005proof} showed that certifying such a spanning tree requires certificates of size \(\Omega(\log n)\). Moreover, this encoding relies on the existence of unique identifiers. In fact, in anonymous
networks, even verifying whether a marked set of edges forms a spanning tree is impossible with the sole help of a prover~\cite{korman2005proof}.

Several variants of the local certification model have been considered. For example, one may allow nodes to inspect certificates at some fixed constant radius, or to access not only the certificates but also the input states of their neighbors~\cite{goos2011locally}. Even in such strengthened variants, the certification of a spanning tree appears to require unique identifiers and certificates of logarithmic size.

\subsection{Alternation hierarchies}

A natural way to generalize the certification framework is to allow the nodes to interact with more than one global entity. If all these entities have the same goal, then this does not really change the model: their certificates can simply
be concatenated and viewed as a single certificate provided by only one entity. The interesting case arises when the entities have opposite objectives.

This leads to an alternating model of local certification. Besides the prover, whose goal is to convince the nodes that the input instance is legal, a new player called the \emph{disprover} is introduced. The goal of the disprover is to convince the nodes in the graph that the instance
is illegal. The prover and the disprover take turns, which we call
\emph{interaction rounds}. In each interaction round, one of the two players assigns a label to every node in the network: labels assigned by the prover are called \emph{certificates}, while labels assigned by the disprover are called \emph{challenges}. After these interaction rounds, the verifier remains local: each node inspects only a constant-radius neighborhood, together with the labels provided by the players, and decides whether to accept or reject. As usual, the whole network accepts if and only if all nodes accept.

The number of interaction rounds, together with the order in which the prover
and the disprover provide their labels, induces a hierarchy of models. For
instance, the usual PLS model corresponds to a single \emph{existential} interaction
round: there exists a certificate assignment that makes all nodes accept on
legal instances, while no certificate assignment can make all nodes accept on
illegal instances.

In general, the semantics of this hierarchy is the following: each action of the disprover should be understood as a
\emph{universal} choice, while each action of the prover should be understood as an
\emph{existential} choice. Thus, when the disprover provides a challenge, the protocol
must be correct for every possible challenge; when the prover provides a
certificate, it is enough that there exists a certificate that successfully
answers the previous challenges. In this way, an alternating protocol naturally
gives rise to an alternating sequence of universal and existential choices,
while the final verification step remains purely local.

Let us illustrate this with the property of being non-3-colorable. The usual
3-colorability property has a simple local certificate: the prover assigns one
of three colors to each node, and each node checks that all its neighbors have
a different color. The complement property, namely non-3-colorability, is
naturally expressed with one interaction of the disprover followed by one
interaction of the prover. The disprover first provides a challenge, consisting of an alleged
3-coloring of the graph. If the graph is indeed not 3-colorable, then every
such challenge must fail: there is at least one edge whose endpoints receive
the same color. The prover can then answer the challenge by certifying the
existence of such an invalid edge. To make this witness visible to all nodes,
the prover provides a spanning tree rooted at one endpoint of a monochromatic
edge. The root checks that one of its neighbors has the same color in the
challenge, while the remaining nodes verify the consistency of the spanning
tree, as in a standard PLS certificate. Thus, for every coloring proposed by
the disprover, there exists a certificate that makes all nodes accept. On the other hand, if the graph is 3-colorable, then the disprover can choose
a proper 3-coloring as its challenge. In that case there is no monochromatic edge,
and therefore no spanning-tree certificate rooted at a local defect can make
all nodes accept. 

Several alternation hierarchies for local decision and local certification
have been studied before, under different assumptions on 
certificate size, computational power, and the information available to each
node. In the logarithmic-size hierarchy
\cite{feuilloley2021hierarchy}, nodes have globally unique identifiers and the
players provide \(\mathcal{O}(\log n)\)-bit certificates. This model shows that
alternation can drastically reduce certificate size: for instance, minimum
spanning tree and nontrivial automorphism admit much more succinct protocols
than in one-round certification. However, the logarithmic hierarchy is not
known to be infinite, and it does not contain all computable distributed
languages.

Other hierarchies with global features have also been considered. The
identifier-independent hierarchy of~\cite{balliu2017can} collapses at the
second interaction and already contains all Turing-decidable distributed
languages. A polynomial-time variant was studied in
\cite{aldematshuva_et_al:LIPIcs.OPODIS.2023.27}. These models are very
powerful because certificates or local computations may still encode
information depending on global parameters of the input graph.

The closest point of comparison for our work is the strictly local line of
research. In~\cite{reiter2024local}, the local polynomial and local bounded
hierarchies are introduced in the LOCAL model with locally unique identifiers.
There, certificate sizes and local computation are bounded as functions of the
constant-radius view of each node. This yields a hierarchy that preserves a
strong notion of locality, admits logical characterizations in the spirit of
descriptive complexity, and is known to be infinite. Alternation had also been
studied earlier in finite-state anonymous models, where connections with
monadic second-order logic were established~\cite{Reiter_2015}.

Our work belongs to this strictly local side of the landscape, but focuses on
an even weaker regime: anonymous networks, constant-size labels, and weak
views of neighborhoods. Thus, our goal is not to revisit the previous models with some additional constraints such as smaller certificates, but to understand what
alternation can express when the only information exchanged in each interaction
has constant size and when nodes have only
set, multiset, or
port-numbering views of neighboring certificates. In fact, our work is not focused on the properties of a particular hierarchy of alternation but on understanding  how alternation can impact the power and expressiveness of the model when the communication protocol and the size of the certificates are very limited.

\subsection{Constant-size hierarchies in weak models}

The weak views that we consider are inspired by the weak models of distributed
computing studied in~\cite{hella2012weak}. These models classify anonymous
local computation according to what a node can observe from its neighbors.
In the \emph{port-numbering} view, each incident edge is assigned a distinct
local port number, allowing a node to distinguish the different occurrences
of neighboring labels. In contrast, in the \emph{multiset} view, port
information is lost, but multiplicities are preserved: a node can see how
many neighbors carry each label. In the \emph{set} view, even multiplicities
are lost: a node only sees which labels appear among its neighbors.

The connection with logic provides additional motivation for this distinction.
The multiset and set models were characterized by modal logic
in~\cite{hella2012weak}: roughly speaking, the set view corresponds to
ordinary modal logic, which can express the existence of a neighbor of a
given type, while the multiset view corresponds to graded modal logic,
which can also express counting statements about neighbors. Thus, losing
or preserving multiplicities is not a technical detail; it changes the
expressive power of the local model.

Unless stated otherwise, we consider finite connected
undirected graphs. A network configuration is a pair
\((G,x)\), where \(x:V(G)\to\{0,1\}^*\) assigns an input
to each node. A distributed language is a collection
of such configurations.

We combine these weak views with alternating local certification. This
gives rise to three constant-size hierarchies, depending on whether nodes
observe their neighborhoods as sets, multisets of neighboring label
sequences, or pairs consisting of a local port number and a neighboring
label sequence. We call them the \emph{\seth hierarchy}, the
\emph{\multiseth hierarchy}, and the \emph{\portsh hierarchy},
respectively. These hierarchies are defined in anonymous networks and
with constant-size labels.

The protocol has two phases. First, there is a bounded interaction
between the prover and the disprover: in each interaction round, one of
the two players assigns a constant-size label to every node. Then, once
the interaction is over, the nodes run a local verifier. The verifier is
local in a weak sense: a node can use its own input and labels, but its
information about the neighborhood is limited to the \emph{set} of
neighboring label sequences in the \seth model, the \emph{multiset} of
neighboring label sequences in the \multiseth model, or the set of pairs
consisting of a local port number and a neighboring label sequence in
the \portsh model.

At first sight, these restrictions may seem severe. The players can
provide only constant-size information, the nodes are anonymous, and
the verifier sees only a restricted view of each neighborhood.
Nevertheless, the resulting hierarchies are surprisingly expressive.
They capture many natural properties and exhibit clean connections
with logic. At the same time, the weakness of the local views makes
the models fine-grained enough to separate \seth from \multiseth and
\multiseth from \portsh, to study the role of multiplicities and port
information, and to identify natural properties that remain outside
the hierarchies.

Let us give a rough definition of our models in order to state our
results (see \Cref{sec:def} for more precise definitions).
Given a graph $G=(V,E)$ and a map $x:V\to\{0,1\}^*$, we call the
map $x$ an \emph{input} for $G$, and we call the pair $(G,x)$ a
\emph{network configuration}. A distributed language is a collection
of network configurations.

A \emph{port numbering} of $G$ is a family
$\pi=(\pi_v)_{v\in V}$, where each
$\pi_v:E(v)\to\{1,\ldots,\deg(v)\}$
is a bijection and $E(v)$ denotes the set of edges incident to $v$.
Port numbers are local: the numbers assigned to an edge at its two
endpoints need not coincide. We denote by $\mathcal{P}(G)$ the set of
all valid port numberings of $G$.

We consider three verification models, denoted by
$\modelo\in\{\seth,\multiseth,\portsh\}$, which differ in the
information that a node can access about the certificates assigned
to its neighbors.

For each $k\geq 0$ and each model
$\modelo\in\{\seth,\multiseth,\portsh\}$, we define the
\emph{$k$-th level of the $\modelo$-hierarchy}, denoted
$\hier{\modelo}_k$, as the family of distributed languages $L$
for which there exist a positive integer $c$ and an anonymous
one-round distributed algorithm $A$ such that, for every network
configuration $(G,x)$ and every port numbering
$\pi\in\mathcal{P}(G)$,
\[
\begin{aligned}
(G,x)\in L
\iff{}&
\quantif_1\ell_1,\quantif_2\ell_2,\ldots,
\quantif_k\ell_k
\text{ such that}\\
&\forall v\in V(G),\quad
\alg{
  v, x(v), \ell_1(v), \ldots, \ell_k(v),
  \view{\modelo}{v}
}
= \textsc{accept}.
\end{aligned}
\]
Here, the quantifiers $\quantif_1,\ldots,\quantif_k$ follow
the alternating pattern defining level $k$, and each
$\ell_i:V\to\{0,1\}^*$ satisfies
$|\ell_i(v)|\leq c$ for every $v\in V(G)$.

For each node $u$, let
\(
\overline{\ell}(u)
=
\bigl(\ell_1(u),\ell_2(u),\ldots,\ell_k(u)\bigr)
\)
denote its sequence of certificates over the $k$
interaction rounds. The neighborhood views are defined by
\[
\begin{aligned}
\view{\seth}{v}
&=
\bigl\{
\overline{\ell}(u):u\in N(v)
\bigr\},
\\[2mm]
\view{\multiseth}{v}
&=
\bigl\{\!\bigl\{
\overline{\ell}(u):u\in N(v)
\bigr\}\!\bigr\},
\\[2mm]
\view{\portsh}{v}
&=
\bigl\{
(\pi_v(uv),\overline{\ell}(u)):
u\in N(v)
\bigr\}.
\end{aligned}
\]
Thus, the \seth model discards multiplicities, the
\multiseth model preserves them, and the \portsh model
additionally identifies the local port through which each
certificate sequence is received. In the \portsh model,
the port numbering is fixed before the certificate
assignments are chosen, and the protocol must be correct
for every valid port numbering. The views in the \seth
and \multiseth models are independent of the port numbering.

Finally, we define the \emph{set hierarchy}, the
\emph{multiset hierarchy}, and the \emph{port hierarchy} as
\[
\hier{\seth}
=
\bigcup_{k\geq 0}\hier{\seth}_k,
\qquad
\hier{\multiseth}
=
\bigcup_{k\geq 0}\hier{\multiseth}_k,
\qquad
\hier{\portsh}
=
\bigcup_{k\geq 0}\hier{\portsh}_k.
\]

\subsection{Our results}

We obtain several results on the expressive power of the \seth,
\multiseth, and \portsh hierarchies.
\Cref{fig:hierarchy-inclusions} summarizes their relationships with
MSO and CMSO over finite input alphabets. In particular, increasing
the information available about neighboring certificates yields
strictly more expressive verification models:
\[
    \mathrm{MSO}
    = \hier{\seth}
    \subsetneq \hier{\multiseth}
    \subsetneq \hier{\portsh}.
\]
CMSO also lies strictly between \hier{\seth} and \hier{\portsh}.
Furthermore, we exhibit a language in \hier{\multiseth} that is not
CMSO-definable, while a separation in the reverse direction has not
been established.

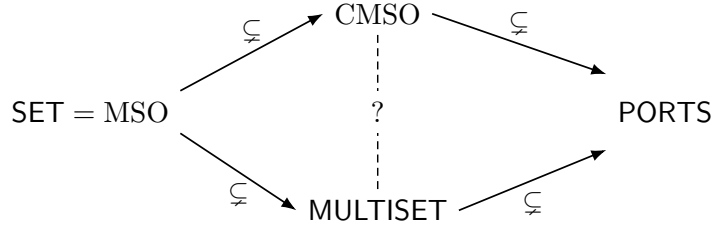
\begin{figure}[t]
\centering
\begin{tikzpicture}[
    class/.style={
        inner sep=4pt,
        font=\normalsize
    },
    inclusion/.style={
        -{Latex[length=2mm]},
        line width=0.65pt
    },
    relation/.style={
        font=\small,
        inner sep=2pt
    },
    unresolved/.style={
        densely dashed,
        line width=0.55pt
    }
]
    \node[class] (set) at (0,0)
        {\(\mathsf{SET}=\mathrm{MSO}\)};

    \node[class] (cmso) at (3.8,1.3)
        {\(\mathrm{CMSO}\)};

    \node[class] (multi) at (3.8,-1.3)
        {\(\mathsf{MULTISET}\)};

    \node[class] (ports) at (7.6,0)
        {\(\mathsf{PORTS}\)};

    \draw[inclusion] (set.north east)
        -- node[relation, midway, above=2pt]
            {\(\subsetneq\)}
        (cmso.west);

    \draw[inclusion] (set.south east)
        -- node[relation, midway, below=2pt]
            {\(\subsetneq\)}
        (multi.west);

    \draw[inclusion] (cmso.east)
        -- node[relation, midway, above=2pt]
            {\(\subsetneq\)}
        ([yshift=5mm]ports.west);

    \draw[inclusion] (multi.east)
        -- node[relation, midway, below=2pt]
            {\(\subsetneq\)}
        ([yshift=-5mm]ports.west);

    \draw[unresolved] (cmso.south)
        -- node[
            midway,
            fill=white,
            inner sep=4pt,
            font=\normalsize
        ] {\(?\)}
        (multi.north);

\end{tikzpicture}
\caption{Inclusions among the verification hierarchies and logics
over finite input alphabets. All solid arrows denote strict inclusion.
The dashed line marks the remaining uncertainty between
\(\mathrm{CMSO}\) and \(\mathsf{MULTISET}\):
we know that
\(\mathsf{MULTISET}\nsubseteq\mathrm{CMSO}\),
whereas the reverse non-inclusion
\(\mathrm{CMSO}\nsubseteq\mathsf{MULTISET}\)
has not been established.}
\label{fig:hierarchy-inclusions}
\end{figure}

\Cref{tab:landscape} complements this overview with concrete languages
and their placement in the \seth and \multiseth hierarchies.

We first focus on the \seth hierarchy. As discussed above, one might expect
this model to be rather weak, since the verification process is highly
restricted. Nevertheless, we show that it is expressive enough to contain
several important languages.

One important distributed language that also plays a key role in this paper is $\leader.$  This language consists of instances $(G = (V,E),x)$ such that $x:V\to\{0,1\}$ and there is a unique node $v$ satisfying  $x(v) =1$. Formally, $$\leader = \{(G,x): x \in \{0,1\}^{V}, \exists! v\in V, x(v) = 1\}.$$ 
The language \leader belongs to the first level of the
logarithmic hierarchy, where \(\mathcal{O}(\log n)\)-bit
certificates suffice. This bound is optimal: every protocol
with a single prover assignment requires
\(\Omega(\log n)\)-bit certificates, even when the verifier
inspects a neighborhood of any fixed constant
radius~\cite{goos2011locally}.
We show that three interaction rounds suffice with
constant-size certificates in the \seth model.

\begin{restatable}{theorem}{thmleader}
\label{cor:leader}
\textnormal{\leader} is in the third level of the \seth hierarchy.
\end{restatable}

For the proof of \Cref{cor:leader}, we develop a technique to certify the
existence of rooted spanning trees despite having no identifiers and only
constant local communication. This technique relies on two main ingredients:
coloring orientations, and acyclicity testing through the surjectivity of
finite-state dynamics of the XOR map. It is a key component in proving the
subsequent results of the paper, as it allows us to show that the hierarchy is
closed under intersection, union, and complement.

This closure phenomenon is perhaps surprising in such a restricted anonymous
setting. Indeed, it means that the \seth hierarchy is not merely a collection of
ad hoc locally certifiable properties, but a robust formalism with the full
Boolean structure. As a consequence,
we obtain that the \seth hierarchy defines exactly the languages definable in Monadic Second-Order Logic (MSO). 

\begin{restatable}{theorem}{thmSETMSO}
\label{thm:set-equals-mso}
    Over finite input alphabets, the \seth hierarchy defines exactly the
    MSO-definable distributed languages.
\end{restatable}

In other words, over finite input alphabets, the expressive power obtained by
constant-size alternating certificates in the anonymous SET model coincides
exactly with MSO definability. Thus, every MSO-definable distributed language
admits a certificate system at some finite level of the \seth hierarchy, and
conversely every language recognized by the hierarchy has an MSO definition.

A direct consequence of \Cref{thm:set-equals-mso} is that
(non-)\(k\)-colorability for each fixed \(k\), planarity, and
\(H\)-minor-freeness for each fixed graph \(H\) belong to
the \seth hierarchy.

We next consider counting predicates. For every fixed modulus
\(p\geq1\) and residue \(0\leq q<p\), we show that the property
of having \(q\) selected vertices modulo \(p\) belongs to the
\portsh hierarchy. The construction uses local port numbers
to select a parent at each non-root vertex and verify modular
sums over the resulting tree. This yields the following
extension of our MSO containment.

\begin{restatable}{theorem}{thmCMSO}\label{th:CMSO}
Every CMSO-definable language belongs to the \portsh hierarchy.
\end{restatable}

We also exhibit a graph language in
\(\hier{\multiseth}_1\) that is not CMSO-definable.
Together with the preceding theorem, this implies that CMSO
is strictly contained in the \portsh hierarchy.
Whether every CMSO-definable language belongs to the
\multiseth hierarchy remains open.

\begin{restatable}{theorem}{thmExout}\label{th:examout}
    There is a language in ${\multiseth}$  that is not CMSO-definable.
\end{restatable}

We next compare the \multiseth and \portsh hierarchies.
The inclusion
\(\hier{\multiseth}\subseteq\hier{\portsh}\)
follows by discarding port numbers. We prove that this
inclusion is strict, even for graph languages without
input labels.

\begin{restatable}{theorem}{thmMultisetPorts}
\label{th:multiset-ports}
There exists a decidable language of 
graphs, without input labels, that belongs to
\(\hier{\portsh}\) but not to \(\hier{\multiseth}\).
Consequently,
\(\hier{\multiseth}\subsetneq\hier{\portsh}\).
\end{restatable}

The proof, given in \Cref{sec:multiset-ports}, encodes binary
words as graphs. On this encoding, \portsh can recognize
every decidable word language closed under reversal.
We then use a counting argument over finite tables of local
views to construct such a word language whose graph encoding
cannot be recognized by any \multiseth protocol.

These results suggest that the expressive power of the \portsh hierarchy can be compared to that of the logarithmic hierarchy. However, the constant
size of the certificates in our model, together with the absence of unique identifiers, also makes it possible to prove separation results whose
analogues remain out of reach for the logarithmic hierarchy. More precisely, we show that the levels of the constant hierarchy form a strict hierarchy.

\begin{restatable}{theorem}{thmHstrict}\label{th:hstrict}
    Let ${\modelo\in\{\seth,\multiseth, \portsh\}}$. The $\modelo$-hierarchy consists of infinitely many levels. More precisely, for all $k\geq0$ there exists a language $L \in \hier{\modelo}_{k+1}$ such that $L\not \in \hier{\modelo}_{k}.$
\end{restatable}

The proof compares corresponding levels of our hierarchies
with those of the locally polynomial hierarchy on graphs of
bounded degree and over finite input alphabets.
For every \(k\geq2\), this comparison transfers the separations
of~\cite[Theorem~36 and Proposition~38]{reiter2024local},
yielding a language in \(\hier{\seth}_{k+1}\) that does not
belong to \(\hier{\portsh}_k\).
The first two inclusions are separated directly using
\(2\)-colorability and \amos, respectively.
Since
\(\hier{\seth}_k\subseteq\hier{\multiseth}_k
\subseteq\hier{\portsh}_k\)
at every level, the same witnesses establish strictness
for all three hierarchies.

Our containment results for MSO and CMSO concern
\(\mathrm{MSO}_1\) and \(\mathrm{CMSO}_1\), respectively,
where second-order quantification ranges over sets of
vertices. We show that the containment in the \portsh
hierarchy cannot be extended to \(\mathrm{MSO}_2\), which
also allows quantification over sets of edges.

\begin{restatable}{theorem}{thmPortsHamPM}
\label{th:ports-ham-pm}
Neither Hamiltonicity nor the existence of a perfect matching
belongs to the \portsh hierarchy. Consequently,
\(\mathrm{MSO}_2\not\subseteq\hier{\portsh}.\)
\end{restatable}

Both properties are definable in \(\mathrm{MSO}_2\).
For Hamiltonicity, one quantifies an edge set forming a
connected spanning subgraph in which every vertex has
degree two. For a perfect matching, one quantifies an edge
set containing exactly one edge incident to each vertex.
Thus, the failure of the containment already occurs for
graph languages without input labels.

The proof, given in \Cref{sec:knnout}, shows that a \portsh
protocol on complete bipartite graphs induces a deterministic
communication protocol of uniformly bounded cost between
the two parts. For either property, such a protocol would
decide whether the two parts have the same size, which
requires unbounded communication.

For languages with unbounded input alphabets, we prove
a further limitation. Let \(\size\) be the language
consisting of all pairs \((G,x)\) such that \(G\) has
\(n\) vertices and \(x(v)=n\) for every \(v\in V(G)\).

\begin{restatable}{theorem}{thmSize}\label{th:langout}
    $\textnormal{\size}$ is not in the $\portsh$ hierarchy. 
\end{restatable}
In fact, our result holds for a much larger set of languages.  We say that a word $s \in \{0,1\}^*$ is $d$-\emph{fragile} if there exists a yes-instance $(G,x)$ and a node $v$ of degree $d$ in $G$ such that the input of $v$ is $s$, i.e., $x(v)=s$ but any perturbation of the input value of $v$ produces a no-instance. We show that, for every language \(L\in\hier{\portsh}\)
and every fixed degree \(d>0\), only finitely many input
values are \(d\)-fragile. For \size, the cycles \(C_n\),
with \(n\geq3\) and input \(n\) at every node, give
infinitely many \(2\)-fragile input values. Therefore,
\(\size\notin\hier{\portsh}\).

Finally, \Cref{tab:landscape} provides an overview of our main results, summarizing the expressive power and limitations of the three hierarchies through representative languages and their corresponding levels. The table also highlights the relationships between the models and the questions that remain open.

\begin{table}[t]
\centering
\small
\setlength{\tabcolsep}{3pt}
\resizebox{\columnwidth}{!}{%
\begin{tabular}{
  @{}>{\raggedright\arraybackslash}p{0.43\columnwidth}
  c c c l@{}
}
\toprule
Language & \seth & \multiseth & \portsh & Result \\
\midrule

\(k\)-coloring
& \(\hier{\seth}_1\)
& \(\hier{\multiseth}_1\)
& \(\hier{\portsh}_1\)
& Folklore \\

At least one selected node \((\alos)\)
& \(\hier{\seth}_3\)
& \(\hier{\multiseth}_3\)
& \(\hier{\portsh}_3\)
& \Cref{th:alos} \\

At most one selected node \((\amos)\)
& \(\hier{\seth}_2\)
& \(\hier{\multiseth}_2\)
& \(\hier{\portsh}_2\)
& \Cref{th:amos} \\

Exactly one selected node \((\leader)\)
& \(\hier{\seth}_3\)
& \(\hier{\multiseth}_3\)
& \(\hier{\portsh}_3\)
& \Cref{cor:leader} \\

MSO-definable languages
& \(\subseteq\hier{\seth}\)
& \(\subseteq\hier{\multiseth}\)
& \(\subseteq\hier{\portsh}\)
& \Cref{thm:set-equals-mso} \\

CMSO-definable languages
& \(\not\subseteq\hier{\seth}\)
& ?
& \(\subseteq\hier{\portsh}\)
& \Cref{th:CMSO} \\

Degree-neighbor counting 
& \(\notin\hier{\seth}\)
& \(\hier{\multiseth}_1\)
& \(\hier{\portsh}_1\)
& \Cref{th:examout} \\

Hamiltonicity
& \(\notin\hier{\seth}\)
& \(\notin\hier{\multiseth}\)
& \(\notin\hier{\portsh}\)
& \Cref{th:ports-ham-pm} \\

Existence of a perfect matching
& \(\notin\hier{\seth}\)
& \(\notin\hier{\multiseth}\)
& \(\notin\hier{\portsh}\)
& \Cref{th:ports-ham-pm} \\

Graph size verification \((\size)\)
& \(\notin\hier{\seth}\)
& \(\notin\hier{\multiseth}\)
& \(\notin\hier{\portsh}\)
& \Cref{th:langout} \\

\bottomrule
\end{tabular}%
}

\caption{Examples of languages studied in this paper
and their placement in the \seth, \multiseth, and
\portsh hierarchies. The indicated levels are upper
bounds obtained by our constructions and are not
necessarily optimal. For logical classes, inclusions
refer to all languages definable in the corresponding
logic. ``Open'' denotes an unresolved inclusion.}
\label{tab:landscape}
\end{table}

\subsection{Structure of the article}
The rest of the paper is organized as follows.  In
\Cref{sec:highlevel}, we provide a high-level overview of our results and the
main ideas behind the proofs. The detailed proofs are given in
\Cref{sec:constructions,sec:SETequalsMSO,sec:MSOCMSOinH,sec:Hstrict,sec:sizeout}. We conclude in \Cref{sec:discussion} with a discussion of possible directions for future work.











\section{High-level description of our results} \label{sec:highlevel}

In this section, we provide a high-level description of the ideas and techniques used to obtain our results. 

\subsection{The power of the SET model}

We first focus on proving \Cref{thm:set-equals-mso}, showing that the \seth hierarchy has the same expressive power as MSO. 

We begin developing a series of fundamental tools. First, we show that the language $\leader$ is in the third level of the \seth hierarchy. Remember that this language consists of instances $(G,x)$ such that $x:V\to\{0,1\}$ and  $|M(x)|= 1$, where $M(x)=\{u\in V: x(u)=1\}$. It is known that $\leader$ admits a PLS with certificates of size $\mathcal{O}(\log n)$ \cite{korman2005proof}, but this result heavily relies on node identifiers and logarithmic certificates. Instead, we consider a different approach based on graph colorings and finite-state dynamics on graphs. 

First, we divide the task of showing that $\leader$ is in the hierarchy by considering two auxiliary languages: the language $\alos$, consisting of instances $(G,x)$ such that \emph{at least one node is selected} (i.e., $|M(x)|\geq 1$), and the language $\amos$, consisting of instances $(G,x)$ such that \emph{at most one node is selected} (i.e., $|M(x)|\leq 1$). 

\subsubsection{ALOS is in the third level of the SET hierarchy}

We begin showing that $\alos$ is at the third level of the \seth hierarchy

\begin{restatable}{lemma}{thmALOS}\label{th:alos}
    \textnormal{$\alos$} is in the third level of the \seth  hierarchy.
\end{restatable}

The idea of the proof consists of two ingredients. The first ingredient is the
use of \emph{coloring orientations}. Given an instance \((G,x)\), the prover
assigns a root mark and a \(3\)-coloring \(c:V\to\{0,1,2\}\). A node \(v\)
colored \(c(v)\) considers as possible successors the neighbors with color
\(c(v)-1\mod 3\). The root is required to be a selected node and to have no
successor, while every other node is required to have at least one possible
successor.

In the \seth model, the coloring \(c\) alone does not define a unique successor:
nodes can detect which colors appear among their neighbors, but they cannot
distinguish or count neighbors with the same color. For this reason, in the second round the disprover chooses a set \(R\subseteq V\) of red nodes, and in the third round the prover gives an auxiliary constant-size mark. A red node \(v\) looks at its red successors, that is, red neighbors with color \(c(v)-1\mod 3\). If \(v\) sees a unique red-successor type, then it uses this type as its successor. If it sees no red-successor type, or if the auxiliary marks make two different red-successor types visible, then \(v\) treats itself as having no successor. In this way, each red node locally decides whether to use one
visible red-successor value in its XOR equation or no successor
value at all.

The key remark is that \((G,x)\in\alos\) if and only if there exists a
valid rooted \(3\)-coloring \((c,\rho)\) such that, for every set
\(R\subseteq V\), the graph \(D_R=D(G,c,\rho,R)\) is acyclic. In a
yes-instance, the prover can take a BFS coloring modulo \(3\) rooted at a
selected node. Then every directed edge decreases the distance to the root, so
no such graph \(D_R\) contains a directed cycle. Conversely, if no node is
selected, then every valid coloring forces each node to have at least one
possible successor, and the disprover can choose a red directed cycle in which
each red node has exactly one red successor.

The second ingredient comes from a finite-state dynamics on directed graphs. 
In this context, an element of \(\{0,1\}^V\) is called a \emph{configuration}.
Given a directed graph \(D=(V,E)\), we define the map
\(\xor^D:\{0,1\}^V\to\{0,1\}^V\) by
\[
\xor^D(z)_v=
\left(
z_v+\sum_{w\in N^+(v)}z_w
\right)\mod 2.
\]
We show that, restricted to directed graphs of maximum out-degree at most
\(1\), the map \(\xor^D\) is surjective if and only if \(D\) contains no
directed cycle.

Putting everything together, the protocol for \(\alos\)
has three rounds. First, the prover gives a rooted
\(3\)-coloring \((c,\rho)\). Second, the disprover
chooses a red set \(R\subseteq V\) and a bit labeling
\(y\). Third, the prover gives auxiliary bits \(a\)
and a bit labeling \(z\).

At each red vertex \(v\), the verifier considers the
set of distinct visible red-successor types
$A(v)=\{(a(u),z(u)):u\in S_R(v)\}.$
If \(A(v)=\{(\alpha,b)\}\), the verifier checks
\(y(v)=z(v)+b\pmod 2\). Otherwise, it checks
\(y(v)=z(v)\).

These equations are locally checkable in the \seth
model, but they do not necessarily coincide with
those of \(\xor^{D_R}\): several red successors
may carry the same pair \((a,z)\) and therefore
appear as a single type in the local view.

Nevertheless, when \(R\) is the vertex set of a
shortest directed cycle in the successor relation,
every red vertex has exactly one red successor.
In this case, the equations checked by the
verifier coincide with those of \(\xor^{D_R}\),
regardless of the auxiliary bits. This is the
key observation used in the soundness proof.

If \((G,x)\) is a yes-instance, the prover chooses a BFS
coloring modulo \(3\) rooted at a selected node. Hence every
possible successor is strictly closer to the root. For any
challenge \((R,y)\), the prover can then construct the final
configuration \(z\) in increasing order of distance from the
root, choosing each bit so that the local XOR equation at that
vertex is satisfied.

If \((G,x)\) is a no-instance, the disprover chooses a red
directed cycle on which every red node has exactly one red
successor, and then chooses a configuration with no preimage
under the corresponding XOR map. Hence at least one node must
reject. We refer to \Cref{sec:ALOS} for the full proof.

\begin{remark}
Similar ideas are used in~\cite{reiter2024local}
to establish the expressiveness of the local
second-order logic characterizing the locally
polynomial hierarchy (see Example 4 therein).
Our models, however, are substantially weaker,
and such constructions must be revisited with care.
In particular, Hamiltonicity, which is expressible
in the local second-order logic considered
in~\cite[Example 6]{reiter2024local}, cannot be
recognized even in our strongest model, \portsh,
as established in \Cref{th:ports-ham-pm}.
\end{remark}

\subsubsection{AMOS is in the second level of the SET hierarchy}

Next, in \Cref{sec:AMOS}, we study the problem $\amos$ and show that
it belongs to the second level of the \seth hierarchy.

\begin{restatable}{lemma}{thmAMOS}\label{th:amos}
    \textnormal{$\amos$} is in the second level of the \seth hierarchy.
\end{restatable}

The protocol has quantifier pattern \(\forall\exists\).
The disprover assigns bits \(d(v)\), and the prover responds with
bits \(a(v)\). Each vertex checks that its \(a\)-value agrees with
those of all its neighbors and, if selected, that \(a(v)=d(v)\).
Since \(G\) is connected, \(a\) must be constant. If at most one
vertex is selected, the prover can satisfy every challenge;
otherwise, the disprover assigns different bits to two selected
vertices, forcing rejection. See \Cref{sec:AMOS}.

\subsubsection{MSO is contained in the SET hierarchy}

Now we use the previous constructions to show that every graph
property expressible in monadic second-order logic belongs to the \seth
hierarchy. More precisely, we show the following lemma: 

\begin{restatable}{lemma}{lemMSOinSET}
\label{lem:MSOinSET}
Every MSO-definable distributed language over a finite
input alphabet belongs to the \seth hierarchy.
\end{restatable}

The proof follows the natural interpretation of logical variables as certificate
assignments. A monadic second-order variable \(X\) is represented by a bit
labeling of the nodes: a node \(v\) receives value \(1\) if and only if
\(v\in X\). A first-order variable \(x\) is also represented by a bit labeling,
but in this case the labeling must select exactly one node. This condition is
precisely the language \(\leader\), which we have already shown to belong to the
third level of the \seth hierarchy. Thus, whenever a first-order variable is
introduced, the protocol also verifies that the corresponding labeling is a
valid singleton assignment.

Each logical quantifier is simulated by quantifying over the
corresponding bit labeling. Set variables require no additional
condition. For a first-order variable, however, the labeling must
select exactly one node. This condition has to be treated differently
for existential and universal quantifiers. An existential choice is
accepted only if the labeling is a valid leader labeling, whereas an
invalid labeling chosen by a universal quantifier must be ignored.
Thus, schematically,
\[
\exists x\,\psi
\quad\longrightarrow\quad
\exists \ell_x\,
\bigl(\leader(\ell_x)\land\widehat{\psi}\bigr),
\]
whereas
\[
\forall x\,\psi
\quad\longrightarrow\quad
\forall \ell_x\,
\bigl(\neg\leader(\ell_x)\lor\widehat{\psi}\bigr).
\]
Since \leader and its complement belong to the hierarchy, these
validity conditions can be incorporated using the closure properties
proved above.

It remains to explain how the quantifier-free part of the formula is verified.
The atomic formulas of MSO are of three types:
\[
x=y,\qquad x\sim y,\qquad x\in X.
\]
Each of these can be checked locally once the corresponding variable assignments
are encoded by node labelings. Equality \(x=y\) is verified by checking that the
two singleton labelings mark the same node. Adjacency \(x\sim y\) is verified by
the unique node marked by \(x\), which checks whether it sees a neighbor marked
by \(y\). Membership \(x\in X\) is verified by the node marked by \(x\), which
checks whether it also belongs to the set \(X\).

Boolean combinations of atomic formulas are handled using the closure of the
hierarchy under finite intersections, unions, and complement. Therefore the
whole quantifier-free part of the MSO formula can be verified inside the
\seth hierarchy. Combining this with the simulation of the quantifier prefix
shows that every MSO-definable language belongs to the \seth hierarchy. We refer
to \Cref{sec:MSOinSET} for the complete proof.

\subsubsection{The SET hierarchy is contained in MSO}

We next prove the converse direction: every language in the \seth hierarchy is MSO-definable. 

\begin{restatable}{lemma}{lemSETinMSO}
\label{lem:set-in-mso}
Let \(L\in\hier{\seth}_k\) be a language over a finite input
alphabet. Then \(L\) is MSO-definable.
\end{restatable}

The key observation is that a \seth verifier has
access only to a finite local view: its own input
and certificate sequence, together with the set
of certificate sequences assigned to its neighbors.
Since the input and certificate alphabets are finite,
there are only finitely many such views. Importantly,
the verifier cannot read the inputs of its neighbors,
and the neighborhood view does not preserve multiplicities.

This information can be described in MSO. Each certificate labeling is encoded by a finite tuple of monadic second-order variables, one set for each possible certificate value. These sets are required to form a partition of the vertex set, so that every node receives exactly one certificate in each round. Then, for every possible node type, there is an MSO formula saying that a node has that type.

Similarly, for every finite set \(T\) of
neighboring certificate types, there is
an MSO formula expressing that the set
of certificate sequences appearing among
the neighbors of a node is exactly \(T\).
This formula requires every type in \(T\)
to appear at least once and excludes all
other neighboring certificate types.
Crucially, it does not inspect the input
values of neighboring nodes.

Since the verifier is fixed, there are only finitely many accepting  local views. Hence the condition that a node accepts is a finite disjunction of MSO formulas, one for each accepting view. Finally, the quantifier rounds of the
distributed hierarchy are translated into MSO quantifiers over the corresponding certificate sets: existential certificate rounds become existential monadic quantifiers, and universal certificate rounds become universal monadic
quantifiers.

Thus, the whole statement “there exist/for all certificate assignments, every node accepts” can be translated into an MSO sentence. We refer
to \Cref{sec:SETinMSO} for the complete proof. We obtain that every language in the \seth hierarchy is MSO-definable. Consequently, the \seth hierarchy and MSO have the same expressive power. 

\subsubsection{Trading horizon for interaction}

We also consider fixed-radius versions of \seth. In the
\emph{weak \(r\)-horizon model}, a vertex sees its own input and
certificates together with the set of certificate sequences at
other vertices within distance \(r\). This model is not monotone:
at radius two, every vertex of both \(K_{1,n}\) and \(K_{n+1}\)
sees all other vertices, so stars cannot be recognized, although
they can at radius one.

The \emph{strong \(r\)-horizon model} retains the original
neighborhood view and additionally provides, for each
\(2\leq j\leq r\), the set of certificate sequences in the
radius-\(j\) ball. We denote by \(\sethr{r}{k}\) the class of
languages recognized with \(k\) alternating certificate
assignments in this model. It is monotone in \(r\), and
\(\sethr{1}{k}=\hier{\seth}_k\).

We prove that, in this strong model, larger constant horizons do not add
expressive power to the hierarchy. More precisely, we show the following
theorem:

\begin{restatable}{theorem}{thmSEThorizon}\label{th:set-horizon-collapse}
For every fixed \(r\geq 1\),
$
\sethr{r}{~} = \bigcup_{k\geq0}\sethr{r}{k}
=
\bigcup_{k\geq0}\sethr{1}{k} = \mathsf{SET}.$
\end{restatable}

The prover certifies the sets \(T_j(v)\) of certificate
sequences appearing in \(B_j(v)\), for \(0\leq j\leq r\).
Writing \(\gamma(v)\) for the certificate sequence of \(v\),
the verifier checks \(T_0(v)=\{\gamma(v)\}\) and
\[
T_j(v)=T_0(v)\cup\bigcup_{u\in N(v)}T_{j-1}(u).
\]
These checks use only the ordinary \seth view, since
multiplicities are irrelevant. The certified sets
reconstruct the additional horizon information and can
be merged into the last existential assignment, or
supplied in one additional existential round.

\subsection{The power of the PORTS hierarchy}

Next, we focus on understanding the expressive power of the
\portsh hierarchy. Clearly, this hierarchy contains both the
\seth and \multiseth hierarchies, since a \portsh verifier
can recover the corresponding views by ignoring port numbers
and, if necessary, multiplicities. Therefore, the languages
\(\amos\), \(\alos\), and \(\leader\) also belong to the
\portsh hierarchy. Moreover, the closure properties proved
for the \seth hierarchy, namely closure under finite
intersections, finite unions, and complement, extend directly
to the \portsh hierarchy. In particular, every MSO-definable
language belongs to the \portsh hierarchy.


To go beyond MSO, we introduce an additional ingredient:
the ability to count selected nodes modulo a fixed constant.
This is where port numbering becomes useful. By distinguishing
neighbors through their local ports, nodes can select a unique
parent in a certified spanning tree and verify local aggregation
constraints involving the contributions of all their children.

More precisely, for every fixed \(p\geq 2\) and
\(0\leq q<p\), we define the language \(\lcount{q,p}\)
of instances in which the number of selected nodes is
congruent to \(q\) modulo \(p\), and show that
\(\lcount{q,p}\) belongs to the \portsh hierarchy.

\begin{restatable}{lemma}{thmCount}\label{lemm:count}
For all $q,p\in \mathbb{N}$, with $p\geq 1$, and $0\leq q <p$,    \textnormal{$\lcount{q,p}$} belongs to $\hier{\portsh}_{9}$.
\end{restatable}

For the proof of \Cref{lemm:count}, we combine
\(\leader\) with the coloring orientations used
for \(\alos\). In the \portsh model, each non-root
vertex selects a unique parent among its eligible
neighbors using the smallest local port. We then
certify that this functional orientation is acyclic
and has exactly one root. Consequently, the selected
edges form a rooted spanning tree of \(G\).

This rooted tree allows us to count by collecting information from the leaves
towards the root. Given an instance \((G,x)\), the prover assigns to each node
a value \(s_v\in\{0,\ldots,p-1\}\), intended to represent the number of marked
nodes in the subtree rooted at \(v\), modulo \(p\). Each node \(v\) should then
verify that
\begin{equation}\label{eq:sum}
    s_v \equiv x(v)+\sum_{u\in T_v}s_u \pmod p,
\end{equation}
where \(T_v\) denotes the set of children of \(v\) in \(T\).

This bottom-up aggregation is standard in distributed certification
\cite{feuilloley2021hierarchy,korman2005proof}. In our setting, however, one
must be careful about what is locally visible: every non-root node knows its
unique parent in the certified tree, but a node \(v\) does not identify its
children individually.

To make the verification of \Cref{eq:sum} local, we use two additional
interactions. Once \(T\) and the values \(\{s_v\}_{v\in V}\) are fixed, the
disprover marks a node \(v^*\) at which \Cref{eq:sum} is to be tested. The
prover responds by marking precisely the children of \(v^*\). This is locally
checkable: a node is marked by the prover if and only if its parent is the node
marked by the disprover. Then \(v^*\) uses the multiset of marked neighboring
certificates to compute
\[
x(v^*)+\sum_{u\in T_{v^*}}s_u \pmod p
\]
and checks that this value is equal to \(s_{v^*}\).

If the assignment \(\{s_v\}_{v\in V}\) is incorrect, the disprover chooses a
node \(v^*\) for which \Cref{eq:sum} fails. Then some node rejects: either the
prover marks the children inconsistently, or \(v^*\) receives exactly the
values of its children and detects that the modular equation is false.

It remains to deal with the case where the disprover does not select a unique
node \(v^*\). As before, this can be handled using \(\leader\): the verifier
accepts when the challenge is not unique, and performs the modular check only
when the challenge selects a single node. This is implemented by a Boolean
combination of the previously studied languages.

We obtain that \(\lcount{q,p}\) is recognized at a constant level of the
\portsh hierarchy. In our explicit construction, the resulting Boolean
combination places the language at level \(9\). See \Cref{sec:counting} for
the complete proof.

\subsubsection{The PORTS hierarchy strictly extends CMSO}

Combining the previous techniques, we show that every
language expressible in counting monadic second-order
logic (CMSO) belongs to the \portsh hierarchy.

\thmCMSO*

The proof of \Cref{th:CMSO} follows the same general
approach as \Cref{lem:MSOinSET}, with the additional
simulation of counting predicates using the certification
of \(\lcount{q,p}\) as a subroutine. This allows us to
extend the correspondence between logical quantification
and alternating certificates to CMSO formulas.

We then show that this inclusion is strict by exhibiting
a language in the \multiseth hierarchy that is not
definable in CMSO. Since the \multiseth hierarchy is
contained in the \portsh hierarchy, this also establishes
that the expressive power of \portsh goes beyond CMSO.

\thmExout*

For this separation, we consider the language $L$
consisting of graphs in which every node of degree at
least $3$ has the same number of neighbors of degree
$1$ and degree $2$. This language belongs to the first
level of the \multiseth hierarchy. The prover assigns
to each node its degree truncated at $3$, and the nodes
locally verify these labels. Nodes of degree at least
$3$ additionally check that the numbers of neighbors
labeled $1$ and $2$ coincide. These checks are possible
because the \multiseth model preserves multiplicities.

To prove that $L$ is not definable in CMSO, we show that
any CMSO formula defining $L$ would induce a CMSO
formula over finite words recognizing a non-regular
language. This contradicts the fact that CMSO-definable
word languages are regular~\cite{courcellebook}.

Together, these results establish that CMSO is strictly
contained in the \portsh hierarchy. Whether every
CMSO-definable language also belongs to the \multiseth
hierarchy remains open.

\subsection{Equivalence of hierarchies on bounded-degree graphs}

In \Cref{sec:boundeddegree}, we study the expressive power of
our hierarchies when restricted to bounded-degree graphs. We show
that, from sufficiently high levels onward, the \seth, \multiseth,
and \portsh hierarchies coincide level-wise with the locally
polynomial hierarchy introduced in \cite{reiter2024local},
assuming a fixed finite input alphabet.

The main challenge is to simulate the additional information
available in \cite{reiter2024local}, where nodes are equipped with locally
unique identifiers. On bounded-degree graphs, the identifiers
needed within a constant horizon can be represented using
constant-size labels. The key idea is to certify an appropriate
coloring of a power of the graph, ensuring that nearby nodes
receive distinct colors. These colors allow even a \seth verifier
to recover the distinctions between neighboring nodes that are
normally lost when multiplicities are ignored. The proof uses
alternating certificates to verify the validity of the coloring,
with different constructions depending on whether the protocol
starts with an existential or a universal quantifier.

An important consequence is that our hierarchies have infinitely
many distinct levels. Indeed, \cite{reiter2024local} established strict
separations in the locally polynomial hierarchy even on graphs
of bounded degree. By transferring these separations through
our equivalence result, we conclude that the \seth, \multiseth,
and \portsh hierarchies are all infinite.

\subsection{Explicit problems outside the constant hierarchy}

Finally, in \Cref{sec:sizeout}, we establish explicit
limitations of the \portsh hierarchy. We identify two
different obstacles to verification with constant-size
certificates: the limited information that can be exchanged
in graphs of unbounded degree, and the presence of unbounded
local inputs, even on bounded-degree graphs.


We first show that neither Hamiltonicity nor the existence
of a perfect matching belongs to the \portsh hierarchy.

The proof of \Cref{th:ports-ham-pm} relies on a connection
with communication complexity. More precisely, we show that
every language over a fixed finite input alphabet belonging
to the \portsh hierarchy induces a uniformly bounded
deterministic communication protocol on complete bipartite
graphs.

To obtain the separation, we consider complete bipartite
graphs $K_{n,m}$, where Alice knows $n$ and Bob knows $m$.
Such a graph admits a perfect matching if and only if
$n=m$, and a Hamiltonian cycle if and only if $n=m\geq 2$.
However, deciding equality between two arbitrarily large
integers requires an unbounded amount of deterministic
communication. This contradiction shows that neither
language belongs to the \portsh hierarchy.


We then consider a different limitation, arising from
unbounded input alphabets. We define the language $\size$,
consisting of all pairs $(G,x)$ such that $G$ is an
$n$-node graph and $x(v)=n$ for every $v\in V$.

\thmSize*

For the proof of \Cref{th:langout}, we introduce a lemma
based on \emph{fragile inputs}. Let $(G,x)$ be a network
configuration. A \emph{perturbation} at a node $v$
consists of changing its input $x(v)$ to a different value,
while leaving the rest of the configuration unchanged.

For a fixed $d>0$, we say that $L$ has a \emph{fragile
$d$-degree input value} $s$ if there exists an instance
$(G,x)\in L$ and a node $v\in V$ of degree $d$ with
$x(v)=s$, such that every perturbation at $v$ produces
a no-instance. Given a distributed language $L$, we
denote by $\frag{d}{L}$ the number, possibly infinite,
of fragile $d$-degree input values of $L$.

We establish the following lemma.

\begin{restatable}{lemma}{lemlangfrag}
\label{lem:langfrag}
For every language $L$ in the \portsh hierarchy and
every $d>0$, $\frag{d}{L}$ is finite.
\end{restatable}

The proof of \Cref{lem:langfrag} relies on a counting
argument. Since the degree $d$ is fixed and certificates
have constant size, there are only finitely many possible
local certificate views, including port information.
Consequently, infinitely many fragile input values would
imply that two distinct values induce the same local
verification behavior for every possible certificate
view. Replacing one by the other would then preserve
the acceptance of the entire protocol, contradicting
fragility.

As a consequence, $\size$ cannot belong to the \portsh
hierarchy, since it has infinitely many fragile input
values. Indeed, it suffices to consider cycles $C_n$,
where every node has degree $2$ and input $n$. Changing
the input of any node produces a no-instance, showing
that $\frag{2}{\size}$ is infinite. For the complete proofs of these results, refer to
\Cref{sec:sizeout}.

\section{Preliminaries}
\label{sec:prelim}

In this section, we introduce the basic definitions and the distributed
verification models used throughout the paper. We then formally define
the corresponding hierarchies of constant-size certificates and recall
the counting monadic second-order logic on graphs.

\subsection{Basic definitions}
\label{sec:def}

For integers $s<t$, we denote by $[s,t]$ the set of all integers $k$
such that $s\leq k\leq t$. We also use $[t]$ as shorthand for
$[1,t]$, with $[0]=\emptyset$.

Throughout this paper, graphs are finite, simple, and undirected.
Unless stated otherwise, all graphs are connected. Given a graph
$G=(V,E)$, we write $n=|V|$ for its number of nodes and $N(v)$
for the neighborhood of a node $v\in V$. The degree of $v$ is denoted
by $\delta(v)=|N(v)|$, and $\Delta(G)$ denotes the maximum degree
of $G$. We also write $V(G)$ and $E(G)$ for the vertex and edge
sets of $G$, respectively.

\subsection{Distributed languages}

Given a graph $G=(V,E)$, an \emph{input} is a mapping
$x:V\to\{0,1\}^*$. The pair $(G,x)$ is called a
\emph{network configuration}. A \emph{distributed language} is a
set of network configurations, that is,
\[
L\subseteq
\{(G,x):G=(V,E),\ x:V\to\{0,1\}^*\}.
\]
The complement of $L$ is the distributed language
\[
\overline L=\{(G,x):(G,x)\notin L\}.
\]

A \emph{graph language} is a set of graphs.
When treating a graph language as a distributed
language, we fix a singleton input alphabet
\(\{\bot\}\) and identify each graph \(G\)
with the configuration \((G,x_\bot)\), where
\(x_\bot(v)=\bot\) for every \(v\in V(G)\). For example,
the language of $3$-colorable graphs is defined by
\[
3\text{-}\textsf{colorable}
=
\{G:\exists c:V(G)\to[3],\
\forall uv\in E(G),\ c(u)\neq c(v)\}.
\]
We say that a distributed language has a \emph{finite input alphabet}
if the input values of its nodes belong to a fixed finite set.

For each positive integer $\Delta$, let $\mathcal B(\Delta)$ denote
the class of graphs of maximum degree at most $\Delta$, that is,
\[
G\in\mathcal B(\Delta)
\iff
\Delta(G)\leq\Delta.
\]
Given a class of graphs $\mathcal G$ and a distributed language $L$,
we define the restriction of $L$ to $\mathcal G$ as
\[
L_{\mathcal G}
=
\{(G,x)\in L:G\in\mathcal G\}.
\]
For example, restricting $2$-$\textsf{colorable}$ to
$\mathcal B(2)$ yields the language consisting of paths and
even cycles.

\subsection{Local views and verification models}

We consider three anonymous models of local verification, denoted
by \seth, \multiseth, and \portsh. Nodes do not have unique
identifiers, and the models differ in the information they provide
about neighboring certificates.

Given a graph $G=(V,E)$, a \emph{labeling} is a mapping
$\ell:V\to\{0,1\}^*$. A \emph{labeling list} is a tuple
$\overline{\ell}=(\ell_1,\ldots,\ell_k)$ of labelings.
For a positive integer $c$, we say that $\overline{\ell}$ is
\emph{$c$-bounded} if
$
|\overline{\ell}(v)|
=
\sum_{i=1}^{k}|\ell_i(v)|
\leq c
$
for every $v\in V$.

Let $\overline{\ell}$ be a labeling list. We define the following three local views:

\begin{itemize}

\item In the \seth model, a node observes the set of neighboring
labeling lists:
\[
\view{\seth}{v}
=
\{\overline{\ell}(u):u\in N(v)\}.
\]
In particular, multiplicities are not visible.

\item In the \multiseth model, a node observes the multiset of
neighboring labeling lists:
\[
\view{\multiseth}{v}
=
\{\!\{\overline{\ell}(u):u\in N(v)\}\!\}.
\]
Thus, multiplicities are preserved, but neighboring nodes carrying
identical labeling lists cannot be distinguished.

\item In the \portsh model, every node $v$ is equipped with
a local port numbering, given by an injective map
$
\pi_v:N(v)\to[\delta(v)].
$
Port assignments are local: distinct neighbors of a node
receive distinct port numbers, but the assignments at
opposite endpoints of an edge are independent.
The \portsh-view of $v$ is the set of port-certificate pairs
\[
\view{\portsh}{v}
=
\bigl\{
(\pi_v(u),\overline{\ell}(u)):u\in N(v)
\bigr\}.
\]
Hence, a node can distinguish its neighbors through their
local port numbers, even when several neighbors carry
identical labeling lists.

\end{itemize}

The three models satisfy the natural information ordering
\[
\seth\preceq\multiseth\preceq\portsh,
\]
since a set view can be recovered from a multiset view, and a
multiset view can be recovered from a port-ordered tuple.

\paragraph{Local verifiers.}

A \emph{local verification algorithm}, also called a
\emph{distributed verifier}, for model
$\modelo\in\{\seth,\multiseth,\portsh\}$ is a computable
function $A$ which, for every network configuration $(G,x)$,
every labeling list $\overline{\ell}$, and every node $v\in V$,
produces an output
$
\alg{G,x,v,\overline{\ell}}\in\{0,1\}
$
depending only on the tuple
$\bigl(
x(v),
\overline{\ell}(v),
\view{\modelo}{v}
\bigr).$

In particular, a node can read its own input and its own
certificates, but it cannot directly read the inputs of neighboring
nodes. Only their certificate labels are available through the
corresponding local view.

For the \portsh model, the output may depend on the local port
numbering. However, the distributed languages recognized by the
model must be independent of the particular choice of port
numbering.

We interpret output $1$ as \textbf{accept} and output $0$
as \textbf{reject}. When the configuration and labeling list are
clear from the context, we simply write $\alg{v}$ instead of
$\alg{G,x,v,\overline{\ell}}$.

The verifier accepts a configuration with labeling list
$\overline{\ell}$ if every node accepts, that is,
\[
\alg{G,x,\overline{\ell}}=1
\iff
\forall v\in V,\quad
\alg{G,x,v,\overline{\ell}}=1.
\]
Otherwise, we write $\alg{G,x,\overline{\ell}}=0$.

Notice that even with constant-size certificates, the local
verifier need not have a finite input domain, since the original
input $x(v)$ may have arbitrary length. Nevertheless, the set of
possible certificate views is finite in the \seth model. For
\multiseth and \portsh, this is also the case when the degree
is bounded. In particular, port numbering does not provide
additional unbounded information at a fixed degree.

\subsection{Constant-size certification hierarchies}

We now formally define the hierarchy of distributed decision with
constant-size certificates for each model
$\modelo\in\{\seth,\multiseth,\portsh\}$.

Let $k\geq0$. A distributed language $L$ belongs to the $k$-th
level of the hierarchy if there exist a constant $c>0$ and a
distributed verifier $A$ for model $\modelo$ such that
\[
(G,x)\in L
\iff
\quantif_1\ell_1\,
\quantif_2\ell_2\cdots
\quantif_k\ell_k,\quad
\alg{G,x,\overline{\ell}}=1,
\]
where the quantifiers
$\quantif_1,\ldots,\quantif_k\in\{\forall,\exists\}$
alternate, and the labeling list
$\overline{\ell}=(\ell_1,\ldots,\ell_k)$ is $c$-bounded.

Each labeling $\ell_i$ ranges over assignments from a fixed finite
certificate alphabet $C_i$, chosen independently of the graph size.
The alphabets and their encodings are fixed as part of the protocol,
with a total of at most $c$ bits per node across all certificate
assignments.

We distinguish the two possible choices of the first quantifier.
We write
$\hier{\modelo}(\Sigma_k)$
for the class of languages recognized by protocols starting with
an existential quantifier, and
$
\hier{\modelo}(\Pi_k)
$
for those starting with a universal quantifier. In both cases,
the remaining quantifiers alternate. For $k=0$, there are no
certificate assignments, and both classes coincide with local
decision in model $\modelo$.

We write $\hier{\modelo}_k$ for the union of the two classes at
level $k$, and define the full hierarchy by
\[
\hier{\modelo}
=
\bigcup_{k\geq0}\hier{\modelo}_k.
\]

For the \portsh hierarchy, the equivalence defining membership
must hold for every valid port numbering of the graph. Thus, port
numbers may influence the local verification procedure and the
certificates chosen by the provers and disprovers, but not the
language recognized by the protocol.

The information ordering between the three models immediately
implies that, for every $k\geq0$,
\[
\hier{\seth}_k
\subseteq
\hier{\multiseth}_k
\subseteq
\hier{\portsh}_k.
\]
The same inclusions hold separately for the $\Sigma_k$ and
$\Pi_k$ levels.

\subsection{Counting monadic second-order logic (CMSO)}

Finally, we consider graph languages expressible in counting
monadic second-order logic (CMSO), also referred to as
$\mathrm{CMS}_1$~\cite{courcellebook}. This logic extends monadic
second-order logic (MSO) with counting predicates. Both logics
allow quantification over vertices (first-order quantification,
usually denoted by lowercase variables) and over sets of vertices
(monadic second-order quantification, usually denoted by uppercase
variables).\footnote{We do not allow quantification over sets of
edges. Allowing such quantification yields the more expressive
logics $\mathrm{MSO}_2$ and $\mathrm{CMSO}_2$, respectively.}

In addition to Boolean connectives and quantifiers, CMSO formulas
are built from the following atomic formulas, with their
corresponding interpretations:
\begin{itemize}
    \item $x_i=x_j$: equality of vertices;
    \item $x_i\sim x_j$: adjacency of vertices;
    \item $x_i\in X_j$: membership of a vertex in a set;
    \item $\cardmod{q}{p}{X_i}$: the cardinality of $X_i$
    is congruent to $q$ modulo $p$, where $p\geq2$ and
    $0\leq q<p$ are fixed integers.
\end{itemize}
For languages over a fixed finite input alphabet, we also
allow a unary predicate for each input symbol, interpreted
by the vertices carrying that symbol.
MSO is the fragment of CMSO obtained by excluding the counting
predicates $\cardmod{q}{p}{X_i}$. These logics are also commonly
denoted $\mathrm{MSO}_1$ and $\mathrm{CMSO}_1$, respectively,
to emphasize that second-order quantification is restricted to
sets of vertices.

For a CMSO sentence $\phi$, we define the graph language
\[
L_\phi=\{G:G\models\phi\}.
\]
Many natural graph properties can be expressed in MSO. For
instance, the language of $3$-colorable graphs is defined by
the following formula:
\[
\begin{aligned}
\exists X_1\,\exists X_2\,\exists X_3\;
\Bigg(
&\left[
\forall x\,
\bigvee_{i\in[3]} (x\in X_i)
\right]\\
&\wedge
\left[
\bigwedge_{i\in[3]}
\forall x\,\forall y\,
\bigl(
(x\in X_i\wedge y\in X_i)
\rightarrow\neg(x\sim y)
\bigr)
\right]
\Bigg).
\end{aligned}
\]
Other examples of MSO-definable graph properties include
connectivity, planarity, and exclusion of a fixed graph $H$
as a minor.

Every CMSO formula can be transformed into an equivalent formula
in \emph{prenex normal form}, consisting of a sequence of
first-order or second-order quantifiers followed by a
quantifier-free formula:
\[
Q_1z_1\,Q_2z_2\cdots Q_kz_k\;\psi,
\]
where $Q_i\in\{\exists,\forall\}$ and $\psi$ is
quantifier-free.

The \emph{quantifier rank} (or quantifier depth) of a formula
is the maximum nesting depth of its quantifiers. In particular,
a formula in prenex normal form with $k$ quantified variables
has quantifier rank $k$. This should not be confused with
\emph{quantifier alternation depth}, which measures the number
of alternations between existential and universal quantifiers.

In the results relating logical definability to our distributed
hierarchies, we use the number of quantified variables in a
prenex representation as the relevant complexity parameter.
This allows us to simulate each logical quantifier by a
corresponding certificate quantifier, with an additional
constant number of hierarchy levels.


\section{Fundamental Tools} \label{sec:constructions} 
In this section we show the core construction for our results. A key element that we need in our proofs is to be able to verify that at least (resp. at most) one node is marked in a graph. More precisely, we study three languages:
\begin{enumerate}
    \item $\alos = \{(G,x): x:V\mapsto\{0,1\}, |\{v \in V: x(v)=1\}|\geq 1\}.$
    \item $\amos = \{(G,x): x:V\mapsto\{0,1\}, |\{v \in V: x(v)=1\}|\leq 1\}.$
    \item $\leader = \alos \cap \amos. $
\end{enumerate}
The language $\alos $ corresponds to the set of network configurations such that there is \emph{at least one node selected}. In addition,  $\amos $ corresponds to the set of network configurations such that there is \emph{at most one node selected}. Finally, $\leader$ is the set of network configurations in which there is \emph{exactly one node selected}. 
\subsection{Closure under union and intersection}
\label{sec:unionandinter}

In this section, we establish closure under finite
intersections and unions for the \seth, \multiseth,
and \portsh hierarchies. The proof applies to all
three models, so we write
\(\modelo\in\{\seth,\multiseth,\portsh\}\).

We now prove the closure result.

\begin{lemma}\label{lem:intun}
Let \(\modelo\in\{\seth,\multiseth,\portsh\}\). Let \(L_1,\ldots,L_s\) be languages such that \(L_i\in \hier{\modelo}_{k_i}\) for every \(i\in\{1,\ldots,s\}\), and let \(k=\max_i k_i\). Then
\[
\bigcap_{i=1}^s L_i \in \hier{\modelo}_K
\quad\text{and}\quad
\bigcup_{i=1}^s L_i \in \hier{\modelo}_{k+1},
\]
where $K=k$ if there is a single $i$ with $k_i=k$, and $K=k+1$ otherwise.
\end{lemma}

\begin{proof}
  We first prove closure under finite intersection. We give the proof for the case \(s=2\); the generalization to larger $s$ is straightforward.

By adding dummy certificate rounds, we can suppose that both languages are recognized by protocols with exactly $K$ rounds and the same sequence of quantifiers: $K=k$ is enough when there is a single $i$ with $k_i=k$, and  $K=k+1$ is enough otherwise (in particular, when two protocols use $k$ rounds but start with different quantifiers).
We have \(L_1,L_2\in \hier{\modelo}_{K}\).

Thus, for each \(i\in\{1,2\}\), there exists a local verifier \(A_i\) such that \((G,x)\in L_i\) if and only if
\[
\quantif_1 \ell^i_1,\quantif_2 \ell^i_2,\ldots,\quantif_{K} \ell^i_{K},
\forall v\in V,\ 
A_i(G,x,v,\overline{\ell}^i)=1,
\]
where \(\overline{\ell}^i=(\ell^i_1,\ldots,\ell^i_{K})\).

We define a new verifier \(A\) for \(L_1\cap L_2\). In each round \(j\), the certificate assigned to a node is the pair \((\ell^1_j(v),\ell^2_j(v))\). The verifier \(A\) locally simulates both \(A_1\) and \(A_2\) on the corresponding coordinates and accepts at \(v\) if and only if both simulations accept:
\[
A(G,x,v,\overline{\ell}^1,\overline{\ell}^2)
=
A_1(G,x,v,\overline{\ell}^1)
\wedge
A_2(G,x,v,\overline{\ell}^2).
\]

This construction is valid in all three models.
From the neighborhood view of pairs of certificate
sequences, a node can recover the view corresponding
to each coordinate. In the \portsh model, the local
port numbers are preserved in both projections.
Hence the verifier can simulate both original
verifiers.

It follows directly that \((G,x)\in L_1\cap L_2\) if and only if all nodes accept in this combined protocol. Hence \(L_1\cap L_2\in \hier{\modelo}_{K}\). The argument extends to any finite number of languages by taking tuples of certificates and conjoining all verifiers.

We now prove closure under finite union.
Let $L_1,L_2\in\hier{\modelo}_k$.

Consider an alternating quantifier prefix of
length $k+1$ starting existentially:
\[
\exists\,\forall\,\exists\,\cdots.
\]
Every $k$-round protocol starting existentially
can be embedded in the first $k$ rounds of
this prefix. Similarly, every $k$-round
protocol starting universally can be embedded
in rounds $2,\ldots,k+1$.

In the first round, the prover assigns
a selector $b(v)\in\{1,2\}$ to each vertex.
Every vertex checks that its selector agrees
with those of all its neighbors. Since the
graph is connected, acceptance requires
$b$ to be constant throughout the graph.

The remaining certificate fields simulate
the protocol corresponding to the selected
language, using the appropriate positions
in the common quantifier prefix. Unused
certificate fields are ignored.

If $(G,x)\in L_1\cup L_2$, the prover
selects a language containing $(G,x)$
in the first round and follows a winning
strategy for that language. The selected
protocol therefore accepts against every
subsequent universal challenge.

Conversely, suppose that
$(G,x)\notin L_1\cup L_2$.
Any nonconstant selector causes a local
rejection. If the selector is constant,
it specifies a language that does not
contain $(G,x)$. Since the selection was
made before the universal challenges,
the disprover can follow a winning
strategy against the selected protocol.
Thus some vertex rejects.

Therefore,
$L_1\cup L_2\in\hier{\modelo}_{k+1}.$ The same construction works for any fixed finite union, using a constant-size
selector.

All certificates have constant size,
and the argument applies to the \seth,
\multiseth, and \portsh models.
\end{proof}

\subsection{Coloring orientations}

First, we show an auxiliary result regarding orientations. Given a graph \(G\), we would like to encode, using constant-size certificates, a local orientation induced by a coloring modulo \(3\).

Let \(G=(V,E)\) be a graph and let \(x:V\to\{0,1\}\) be an input function. Given a map \(c:V\to\{0,1,2\}\), we define the set of \emph{successors} of a node \(v\in V\), with respect to the coloring~\(c\), as
\[
S(v)=\{u\in N(v): c(u)=c(v)-1 \mod 3\}.
\]

The prover will also mark some nodes as roots. Let \(\rho:V\to\{0,1\}\) be this root marking, and let \(B(\rho)=\{v\in V:\rho(v)=1\}\). We say that the pair \((c,\rho)\) is a \emph{valid rooted coloring} for \((G,x)\) if the following conditions hold.

\begin{itemize}
    \item For every \(v\in B(\rho)\), we have \(x(v)=1\) and \(S(v)=\emptyset\).
    \item For every \(v\notin B(\rho)\), we have \(S(v)\neq\emptyset\).
\end{itemize}

Observe that these conditions are locally checkable in the \seth model, since a node only needs to know whether the set of certificates received from its neighbors contains a node with color \(c(v)-1 \mod 3\).

Later, the disprover will choose a set \(R\subseteq V\) of \emph{red nodes}. For every node \(v\), define \(S_R(v)=S(v)\cap R\). Intuitively, \(S_R(v)\) is the set of red successors of \(v\). Given a valid rooted coloring \((c,\rho)\) and a set \(R\subseteq V\), we define a directed graph \(D_R=D(G,c,\rho,R)\) as follows.

\begin{itemize}
    \item The vertex set of \(D_R\) is \(R\).
    \item For every \(v\in R\), if \(S_R(v)=\{u\}\), then \(D_R\) contains the directed edge \((v,u)\).
    \item If \(S_R(v)=\emptyset\) or \(|S_R(v)|>1\), then \(v\) has no outgoing edge in \(D_R\).
\end{itemize}

Thus, every node in \(D_R\) has out-degree at most \(1\).

\begin{lemma}\label{lem:dycycle}
    Let \((G,x)\) be a network configuration with \(x:V\to\{0,1\}\). Then \((G,x)\in\alos\) if and only if there exists a valid rooted coloring \((c,\rho)\) such that, for every set \(R\subseteq V\), the graph \(D_R=D(G,c,\rho,R)\) has no directed cycle.
\end{lemma}

\begin{proof}
First, suppose that \((G,x)\notin\alos\). Then no node is selected, that is, \(x(v)=0\) for every \(v\in V\). Let \((c,\rho)\) be any valid rooted coloring. Since every root must satisfy \(x(v)=1\), we have \(B(\rho)=\emptyset\). Therefore, by validity, every node \(v\in V\) satisfies \(S(v)\neq\emptyset\).

Now consider the directed relation that connects each node \(v\) to all nodes in \(S(v)\). Since every node has at least one successor, this relation contains a directed cycle, and we can consider a minimal such cycle. More precisely, there exist distinct nodes \(R=\{v_0,\dots,v_{k-1}\}\) such that \(S_R(v_i)= \{v_{i+1}\}\) for every \(i\), with indices taken modulo \(k\). Indeed, if \(v_j\in S(v_i)\) for \(j\neq i+1\), we could form a cycle avoiding \(v_{i+1}\) by connecting \(v_i\) to \(v_j\) instead of \(v_{i+1}\) (taking indices modulo \(k\)), which would contradict minimality of cycle $R$. Hence \(D_R\) contains the directed cycle \((v_0,v_1,\dots,v_{k-1},v_0)\).

Conversely, suppose that \((G,x)\in\alos\). Let \(r\) be a selected node, and let \(T\) be a BFS tree of \(G\) rooted at \(r\). Define \(\rho(r)=1\), \(\rho(v)=0\) for every \(v\neq r\), and \(c(v)=d_T(r,v)\mod 3\), where \(d_T(r,v)\) denotes the distance from \(r\) to \(v\) in \(T\).

We claim that \((c,\rho)\) is a valid rooted coloring. Since \(r\) is selected, \(x(r)=1\). Moreover, \(S(r)=\emptyset\), because no neighbor of the root has distance \(-1\) from \(r\). For every node \(v\neq r\), the parent of \(v\) in the BFS tree has distance \(d_T(r,v)-1\), and therefore belongs to \(S(v)\). Hence \(S(v)\neq\emptyset\).

Now let \(R\subseteq V\) be arbitrary. We prove that \(D_R\) is acyclic. If \((v,u)\) is an edge of \(D_R\), then \(u\in S(v)\). Since the coloring comes from BFS distances, this implies \(d_T(r,u)=d_T(r,v)-1\). Therefore, the distance to the root strictly decreases along every directed edge of \(D_R\). Hence \(D_R\) cannot contain a directed cycle.
\end{proof}

In the next sections, we use this fact to verify that \(\alos\) belongs to the hierarchy. The remaining difficulty is to express, in the \seth model, the XOR verification associated with such red orientations, without relying on ports or on multiplicities of identical neighboring certificates.

\subsection{Surjectivity of XOR maps}

Given a directed graph \(D=(V,E)\), we define the \emph{XOR map over \(D\)}, denoted by \(\xor^D:\{0,1\}^V\to\{0,1\}^V\), as follows:
\[
\xor^D(z)_v = z_v + \sum_{u\in N_D^+(v)} z_u \mod 2.
\]

We sometimes refer to maps \(z:V\to\{0,1\}\) as \emph{state configurations}, or simply as configurations.

In this paper we will apply this map only to directed graphs of maximum out-degree at most \(1\). For such a graph \(D\), we define the function \(\parent_D:V\to V\cup\{\bot\}\) by setting \(\parent_D(v)=u\) if \(N_D^+(v)=\{u\}\), and \(\parent_D(v)=\bot\) if \(N_D^+(v)=\emptyset\). Then, for every node \(v\), we have
\[
\xor^D(z)_v =
\begin{cases}
    z_v & \text{if } \parent_D(v)=\bot,\\
    z_v+z_{\parent_D(v)} \mod 2 & \text{otherwise.}
\end{cases}
\]

For the sake of simplicity, we write \(\xor\) instead of \(\xor^D\) whenever the underlying directed graph is clear from the context.

The following lemma links the surjectivity of the XOR map with the acyclicity of the underlying directed graph.

\begin{lemma}\label{lemm:XORsurj}
    Let \(D=(V,E)\) be a directed graph with maximum out-degree at most \(1\). Then \(\xor^D\) is surjective if and only if \(D\) contains no directed cycle.
\end{lemma}

\begin{proof}
First, suppose that \(D\) contains no directed cycle. Since every node has out-degree at most \(1\), every directed path eventually reaches a node of out-degree \(0\). Thus, \(D\) is a collection of directed trees whose edges point toward roots of out-degree \(0\).

Let \(y\in\{0,1\}^V\) be arbitrary. We construct a configuration \(z\in\{0,1\}^V\) such that \(\xor(z)=y\).

Let \(B_0\) be the set of nodes of out-degree \(0\). For every \(v\in B_0\), set \(z_v=y_v\). Now suppose that \(z\) has already been defined for all nodes at directed distance at most \(i\) from \(B_0\). Let \(v\) be a node at directed distance \(i+1\) from \(B_0\), and let \(u=\parent_D(v)\). Then \(z_u\) has already been defined, so we set \(z_v=y_v+z_u\mod 2\). This guarantees that \(\xor(z)_v=y_v\).

Proceeding inductively, we define \(z\) for every node of \(D\), and therefore \(\xor\) is surjective.

Conversely, suppose that \(D\) contains a directed cycle \(C\). Let \(A(C)\) be the set of nodes that can reach \(C\) by following directed edges. That is, \(v\in A(C)\) if and only if there exists a directed path from \(v\) to some node of \(C\).

Let \(\mathbf{0}\) be the all-zero configuration, and define a configuration \(y\in\{0,1\}^V\) by setting \(y_v=1\) if \(v\in A(C)\), and \(y_v=0\) otherwise.

Clearly \(\xor(\mathbf{0})=\mathbf{0}\). We claim that also \(\xor(y)=\mathbf{0}\). Indeed, if \(v\in A(C)\), then \(v\) has exactly one outgoing neighbor \(u\), and this neighbor also belongs to \(A(C)\). Hence \(\xor(y)_v=1+1=0\mod 2\). If \(v\notin A(C)\), then either \(v\) has no outgoing neighbor, in which case \(\xor(y)_v=y_v=0\), or its outgoing neighbor also lies outside \(A(C)\), in which case again \(\xor(y)_v=0+0=0\mod 2\).

Thus \(\xor(\mathbf{0})=\xor(y)\), while \(y\neq \mathbf{0}\). Therefore, \(\xor\) is not injective. Since \(\{0,1\}^V\) is finite, \(\xor\) is not surjective.
\end{proof}

\subsection{ALOS is in the third level of the \seth hierarchy}
\label{sec:ALOS}

Now, we put all the pieces together and show that \(\alos\) is in the third level of the \seth hierarchy. We start by stating the following immediate consequence of \Cref{lem:dycycle} and \Cref{lemm:XORsurj}.

\begin{lemma}\label{lemm:nonempt}
    Let \(G=(V,E)\) be a graph and let \(x:V\to\{0,1\}\) be an input function. Then \((G,x)\in\alos\) if and only if there exists a valid rooted coloring \((c,\rho)\) such that, for every set \(R\subseteq V\), the \(\xor\) map defined over \(D_R=D(G,c,\rho,R)\) is surjective.
\end{lemma}

\begin{proof}
    By \Cref{lem:dycycle}, \((G,x)\in\alos\) if and only if there exists a valid rooted coloring \((c,\rho)\) such that, for every set \(R\subseteq V\), the graph \(D_R=D(G,c,\rho,R)\) has no directed cycle. By \Cref{lemm:XORsurj}, the map \(\xor\) over \(D_R\) is surjective if and only if \(D_R\) contains no directed cycle.
\end{proof}

\Cref{lemm:nonempt} characterizes \(\alos\) in
terms of the surjectivity of XOR maps on the
graphs \(D_R\). However, a \seth verifier cannot
directly check the equations of \(\xor^{D_R}\),
since determining whether a vertex has exactly
one red successor requires distinguishing
neighbors that may carry identical certificates.

Instead, our protocol uses locally checkable
XOR-style equations determined by the distinct
certificate types visible among red successors.
We prove completeness by showing that, for a
BFS coloring rooted at a selected vertex,
the prover can satisfy these equations for
every challenge.

For soundness, the disprover chooses the vertex
set of a shortest directed cycle in the successor
relation. Every red vertex then has exactly one
red successor, so the local equations coincide
with those of \(\xor^{D_R}\). The non-surjectivity
of this map provides a challenge that no prover
response can satisfy.

\thmALOS*

\begin{proof}
The protocol has three rounds. In the first round, the prover assigns to every node a rooted coloring \((c,\rho)\), where \(c(v)\in\{0,1,2\}\) and \(\rho(v)\in\{0,1\}\).

In the second round, the disprover assigns to every node a red bit \(r(v)\in\{0,1\}\), defining the set \(R=\{v\in V:r(v)=1\}\), and a bit \(y(v)\in\{0,1\}\). The bit \(y(v)\) is interpreted as a challenge configuration on the red nodes.

In the third round, the prover assigns to every node a bit \(z(v)\in\{0,1\}\) and an auxiliary bit \(a(v)\in\{0,1\}\).

Each node \(v\) first checks that \(x(v)\in\{0,1\}\),
rejecting otherwise. It then verifies that \((c,\rho)\)
is a valid rooted coloring:
\begin{enumerate}
    \item if \(\rho(v)=1\), then \(v\) checks that
    \(x(v)=1\) and \(S(v)=\emptyset\);
    \item if \(\rho(v)=0\), then \(v\) checks that
    \(S(v)\neq\emptyset\).
\end{enumerate}
These checks are valid in \seth: the input conditions
are local, and the successor conditions only require
testing the presence of neighboring colors.

If \(v\notin R\), then \(v\) performs no further check. Suppose then that \(v\in R\). Let \(A(v)\) be the set of pairs \((a(u),z(u))\) that appear among red successors of \(v\), that is, among nodes \(u\in S_R(v)\). 

The verifier now checks the local XOR equation as follows:
\begin{itemize}
    \item if \(A(v)=\emptyset\), then \(v\) checks that \(y(v)=z(v)\);
    \item if \(A(v)=\{(\alpha,b)\}\), then \(v\) checks that \(y(v)=z(v)+b\mod 2\);
    \item if \(|A(v)|>1\), then \(v\) checks that \(y(v)=z(v)\).
\end{itemize}

In other words, if \(v\) sees a unique red-successor type, it uses that type as its successor in the XOR equation. If it sees either no red-successor type or at least two different red-successor types, it treats itself as having no successor. The auxiliary bit is precisely what allows the last prover to exhibit that two red successors should be distinguished, in which case the outgoing edge of \(v\) is cut.

We now prove soundness and completeness.

First, assume that \((G,x)\in\alos\). Let \(r\) be a selected
node. In the first round, the prover chooses \(\rho(r)=1\),
\(\rho(v)=0\) for every \(v\neq r\), and a BFS coloring $
    c(v)=\dist(r,v)\bmod 3.$
This is a valid rooted coloring. Moreover, for every non-root
vertex \(v\),
\[
    S(v)=\{u\in N(v):\dist(r,u)=\dist(r,v)-1\}.
\]
In particular, every successor of \(v\) is strictly closer
to \(r\).

Now consider an arbitrary challenge \((R,y)\) chosen by the
disprover. In the third round, the prover sets $a(v)=0$ for every $v\in V$,
and constructs \(z\) in increasing order of distance from
\(r\).

Suppose that the values \(z(u)\) have already been defined
for all vertices \(u\) with
\(\dist(r,u)<\dist(r,v)\). Since every successor of \(v\)
is closer to \(r\), all values \(z(u)\) for
\(u\in S_R(v)\) are already fixed. If \(v\notin R\), choose \(z(v)\) arbitrarily. If \(v\in R\),
define
\[
z(v)=
\begin{cases}
y(v),
&
\text{if }
\{z(u):u\in S_R(v)\}
\text{ is empty or contains both bits},\\[1mm]
y(v)+b\pmod 2,
&
\text{if }
\{z(u):u\in S_R(v)\}=\{b\}.
\end{cases}
\]

Since \(a\equiv0\), the set \(A(v)\) seen by a red vertex is $ A(v)=\{(0,z(u)):u\in S_R(v)\}.$
Therefore, if the successor values are empty or contain both
\(0\) and \(1\), the verifier uses the equation
\(y(v)=z(v)\); if all red successors have the same value
\(b\), it uses the equation
\(y(v)=z(v)+b\pmod2\).
By construction, the corresponding equation is satisfied
at every red vertex.

Hence, for every challenge \((R,y)\), the prover has a
third-round response \((a,z)\) that makes all vertices
accept.

Now assume that \((G,x)\notin\alos\). If some vertex
has input outside \(\{0,1\}\), it rejects regardless
of the certificates. Otherwise, \(x\in\{0,1\}^V\),
so no vertex is selected. Consider any first-round
certificate \((c,\rho)\) proposed by the prover.

If \((c,\rho)\) is not a valid rooted coloring, then at least one node rejects in the  verification step, independently of the remaining certificates.

Suppose then that \((c,\rho)\) is valid. By \Cref{lem:dycycle}, there exists a set \(R\subseteq V\) such that \(D_R=D(G,c,\rho,R)\) contains a directed cycle. The disprover chooses \(R\) to be the vertex set of such a cycle, so that every red node has exactly one red successor.

By \Cref{lemm:XORsurj}, the map \(\xor^{D_R}\) is not surjective. Hence there
exists a configuration \(y:R\to\{0,1\}\) with no preimage under \(\xor^{D_R}\).
The disprover chooses this configuration \(y\).

Now consider any third-round answer of the prover, including any auxiliary
bits and any configuration \(z\). Since every red node has exactly one red
successor, the auxiliary bits do not affect the local equation checked by the
verifier: each red node \(v\) checks precisely
\(y(v)=z(v)+z(\parent_{D_R}(v))\mod 2\).
Thus, if all red nodes accepted, we would have \(\xor^{D_R}(z)=y\), contradicting
the choice of \(y\). Therefore at least one red node rejects.

All certificates have constant size: \(c(v)\), \(\rho(v)\), \(r(v)\), \(y(v)\), \(z(v)\), and \(a(v)\) are all taken from fixed finite alphabets. We conclude that \(\alos\) belongs to the third level of the \seth hierarchy.
\end{proof} 

\subsection{AMOS is in the second level of the \seth hierarchy}
\label{sec:AMOS}

In this section, we show that \(\amos\) belongs to the
second level of the \seth hierarchy. Let \(G=(V,E)\)
be a graph and let \(x:V\to\{0,1\}\) be an input
function. We write \(M(x)=\{v\in V:x(v)=1\}\).
The language \(\amos\) consists of all configurations
\((G,x)\) such that \(|M(x)|\leq 1\).

The idea is to use a universal challenge to force two
marked nodes, if they exist, to disagree on a globally
constant labeling. Since the graph is connected, the
prover can enforce such a labeling using only local
equality checks.

\thmAMOS*

\begin{proof}
We describe a two-round protocol in the \seth model.

In the first round, the disprover assigns a bit
\(d(v)\in\{0,1\}\) to every node \(v\in V\).
In the second round, the prover answers with a bit
\(a(v)\in\{0,1\}\) at every node.

The local verifier at each node \(v\) first checks
that \(x(v)\in\{0,1\}\), rejecting otherwise.
It then checks the following conditions:

\begin{enumerate}
    \item For every neighbor \(u\in N(v)\),
    $
        a(v)=a(u).
    $
    \item If \(x(v)=1\), then
    $
        a(v)=d(v).
    $
\end{enumerate}

Both conditions are locally checkable in the \seth
model. Indeed, a node knows its own input and
certificates, and can verify whether all neighboring
certificates have the same \(a\)-value as its own.

We now prove correctness.

Suppose first that \((G,x)\in\amos\).
Let \(d:V\to\{0,1\}\) be an arbitrary challenge
of the disprover.

If \(M(x)=\emptyset\), the prover assigns
\(a(v)=0\) to every node. Both conditions are
satisfied.

If \(M(x)=\{r\}\), the prover assigns
\(a(v)=d(r)\) to every node \(v\in V\).
The first condition holds because \(a\) is
constant, and the second holds because \(r\)
is the only marked node.

Thus, for every challenge \(d\), there exists
a response \(a\) such that all nodes accept.

Conversely, suppose that \((G,x)\notin\amos\).
If some vertex has input outside \(\{0,1\}\),
it rejects immediately. Otherwise,
\(x\in\{0,1\}^V\), and there exist two
distinct marked nodes \(u,v\in M(x)\).

The disprover chooses a challenge \(d\) such that
\(d(u)=0\) and \(d(v)=1\).

Assume, for contradiction, that the prover can
produce a labeling \(a\) accepted by every node.
By the first verification condition and the
connectivity of \(G\), the labeling \(a\) must
be constant over \(V\). However, the second
condition requires
\[
a(u)=d(u)=0
\quad\text{and}\quad
a(v)=d(v)=1,
\]
a contradiction.

Therefore, for this challenge, every response
of the prover causes at least one node to reject.

Finally, both certificates consist of a single
bit per node. Hence the protocol uses constant-size
certificates and has the quantifier pattern
$\forall d\,\exists a.$
We conclude that $\amos\in\hier{\seth}_2.$
\end{proof}

\subsection{Leader, spanning-tree and closure operations.}\label{sec:leader}

Now we state some corollaries of the previous results. The first important corollary is that the problem $\leader$ is in the hierarchy $\hier{\seth}.$

\thmleader*

\begin{proof}
    Since $\leader = \alos \cap \amos,$ as a consequence of the first result in \Cref{lem:intun} we have that  $\leader \in \hier{\seth}_3.$
\end{proof}

A direct application of the latter corollary and of the coloring orientations used in the proof of \Cref{th:alos} is that a suitable spanning-tree-like language belongs to the \seth hierarchy. 

It is important to notice that we use a slightly unconventional definition of this language. Traditionally, see, e.g.~\cite{korman2005proof}, the language \textsc{SpanningTree} consists of graphs with marked edges, where the marked edges encode a spanning tree \(T\) of \(G\). This definition does not suit our model, since in the \seth model we have no identifiers or ports that would allow nodes to refer to individual incident edges.

Instead, we define a \seth-induced spanning tree language, denoted by \(\spann\), whose instances are network configurations of the form \((G,(x,c))\), where \(x:V\to\{0,1\}\) and \(c:V\to\{0,1,2\}\). The map \(x\) is intended to mark the root, and the coloring \(c\) is intended to define local successor relations through the sets
\(S(v)=\{u\in N(v):c(u)=c(v)-1\mod 3\}\).

Formally, we say that \((G,(x,c))\in\spann\) if the following conditions hold:
\begin{enumerate}
    \item there exists a unique node \(r\in V\) such that \(x(r)=1\);
    \item the pair \((c,x)\) is a valid rooted coloring;
    \item for every set \(R\subseteq V\), the graph \(D_R=D(G,c,x,R)\) is acyclic.
\end{enumerate}

Equivalently,
\[
\spann=\left\{(G,(x,c)):
\begin{array}{l}
(G,x)\in\leader,\\
(c,x)\text{ is a valid rooted coloring,}\\
\text{for every }R\subseteq V,\ D_R\text{ is acyclic}
\end{array}
\right\}.
\]

Observe that the first condition is exactly \((G,x)\in\leader\). The second condition is locally checkable in the \seth model. The third condition can be verified using the XOR characterization of acyclicity from \Cref{lemm:XORsurj}, exactly as in the proof of \Cref{th:alos}.

More precisely, we obtain the following corollary.

\begin{corollary}\label{cor:span}
    \(\textnormal{\spann}\in \hier{\seth}_3\).
\end{corollary}

\begin{proof}
We verify the three conditions in the definition of \(\spann\).

The first condition, \((G,x)\in\leader\), belongs to
\(\hier{\seth}_3\) by \Cref{cor:leader}.

For the remaining conditions, the prover uses a certificate
\(\widehat c:V\to\{0,1,2\}\) copying the input coloring \(c\).
Every vertex checks locally that $
    \widehat c(v)=c(v).$ 
Hence neighboring values of \(c\) become visible through the
certificates. The validity of the rooted coloring \((c,x)\)
can then be checked locally: a vertex with \(x(v)=1\) checks
that it has no successor, while every vertex with \(x(v)=0\)
checks that it has at least one successor.

It remains to verify that, for every \(R\subseteq V\), the
graph \(D_R=D(G,c,x,R)\) is acyclic. Consider the directed
relation $
    H=\{(v,u):u\in S(v)\}.$
We first observe that
\[
    H\text{ is acyclic}
    \quad\Longleftrightarrow\quad
    D_R\text{ is acyclic for every }R\subseteq V.
\]
The forward implication is immediate. Conversely, suppose
that \(H\) contains a directed cycle, and choose one of
minimum length. Let \(R\) be its vertex set. By minimality,
each vertex of the cycle has exactly one successor in \(R\);
otherwise an additional edge between two vertices of the
cycle would produce a shorter directed cycle. Hence this
cycle also appears in \(D_R\).

We now give a two-round protocol for the acyclicity condition.
The disprover chooses \(R\subseteq V\) and a bit labeling
\(y\). The prover responds with a bit labeling \(z\), an
auxiliary bit labeling \(a\), and the copy \(\widehat c\)
described above.

If \(H\) is acyclic, process its vertices in an order in which
every successor of a vertex has already been processed.
Set \(a\equiv0\), and define \(z(v)\) exactly as in the
completeness proof of \Cref{th:alos}: if the set of \(z\)-values
of red successors of \(v\) is empty or contains both bits,
set \(z(v)=y(v)\); if all red successors have the same value
\(b\), set $z(v)=y(v)+b\pmod2.$
Thus every local equation checked by the verifier is satisfied.

Conversely, if some \(D_R\) contains a directed cycle, the
disprover chooses \(R\) to be the vertex set of such a cycle
and chooses a configuration \(y\) with no preimage under
\(\xor^{D_R}\). Every red vertex then has exactly one red
successor, so the auxiliary bits cannot change the equation
being checked. Hence no response of the prover makes all
vertices accept.

Therefore the second and third conditions can be verified
with the quantifier pattern
$\forall (R,y)\;\exists(\widehat c,a,z).$
They belong to \(\hier{\seth}_2\). Intersecting this language
with \(\leader\in\hier{\seth}_3\), using
\Cref{lem:intun}, gives
$\spann\in\hier{\seth}_3.$
\end{proof}

Now we show that the complement of a language in the \seth hierarchy also belongs to the hierarchy, with a constant increase in the level. More precisely, for \(k\geq 0\), let
\[
\textnormal{co-}\hier{\seth}_k=\{L:\overline{L}\in \hier{\seth}_k\}.
\]

\begin{corollary}\label{cor:comp}
    \(\textnormal{co-}\hier{\seth}_k\subseteq \hier{\seth}_{k+3}\).
\end{corollary}

\begin{proof}
Let \(L\in \textnormal{co-}\hier{\seth}_k\). Then \(\overline{L}\in \hier{\seth}_k\). By definition, there exists a \seth local verifier \(A\) such that
\[
(G,x)\in\overline{L}
\iff
\quantif_1\ell_1\,
\quantif_2\ell_2\,
\cdots
\quantif_k\ell_k\,
\forall v\in V,\ 
A(G,x,v,\overline{\ell})=1,
\]
where \(\overline{\ell}=(\ell_1,\ldots,\ell_k)\).

Taking the complement of this statement, we obtain
\[
(G,x)\in L
\iff
\overline{\quantif}_1\ell_1\,
\overline{\quantif}_2\ell_2\,
\cdots
\overline{\quantif}_k\ell_k\,
\exists v\in V,\ 
A(G,x,v,\overline{\ell})=0,
\]
where \(\overline{\exists}=\forall\) and \(\overline{\forall}=\exists\).

The existential quantifier over nodes can be expressed by means of a non-empty set of marked nodes. More precisely, introduce an additional labeling \(\ell':V\to\{0,1\}\). The condition \(\exists v\in V\) such that \(A(G,x,v,\overline{\ell})=0\) is equivalent to the existence of a labeling \(\ell'\) such that:
\begin{enumerate}
    \item \((G,\ell')\in\alos\);
    \item for every node \(v\), if \(\ell'(v)=1\), then \(A(G,x,v,\overline{\ell})=0\).
\end{enumerate}

Therefore,
$$
\begin{aligned}
(G,x)\in L
\iff\;&
\overline{\quantif}_1\ell_1\,
\overline{\quantif}_2\ell_2\,
\cdots
\overline{\quantif}_k\ell_k\,
\exists \ell'\in\{0,1\}^V
\\
&\text{such that }
\begin{cases}
(G,\ell')\in\alos,\\
\forall v\in V,\ \ell'(v)=1\Rightarrow
A(G,x,v,\overline{\ell})=0.
\end{cases}
\end{aligned}
$$

The condition \((G,\ell')\in\alos\) is in \(\hier{\seth}_3\) by \Cref{th:alos}. The second condition is locally checkable without adding any new alternation: each node only has to verify that, if it is marked by \(\ell'\), then the original verifier \(A\) rejects at that node.

Thus, after the \(k\) complemented quantifier rounds, we can append the three-round \seth protocol for \(\alos\). This gives a protocol in \(\hier{\seth}_{k+3}\). Hence \(L\in \hier{\seth}_{k+3}\), and therefore
\[
\textnormal{co-}\hier{\seth}_k\subseteq \hier{\seth}_{k+3}.
\]
\end{proof}

\begin{example}
    Since \(3\)-\textsf{colorability} belongs to \(\hier{\seth}_1\), 
    \Cref{cor:comp} implies that \(\textsf{non}\)-\(3\)-\textsf{colorability} 
    belongs to \(\hier{\seth}_4\). We remark that \cite{goos2011locally} shows 
    that \(\textsf{non}\)-\(3\)-\textsf{colorability} requires certificates of 
    size \(\tilde{\Omega}(n^2)\) at the first level of the hierarchy, even when 
    the nodes of the network have unique identifiers and the verification 
    algorithm is allowed to inspect a constant-radius neighborhood.
\end{example}

The same argument applies to the \multiseth and
\portsh hierarchies. Indeed, the proof only uses
the fact that \(\alos\) belongs to the third level
and that the local rejection condition of the
original verifier can be checked in the underlying
model. Since
\[
    \hier{\seth}_3
    \subseteq\hier{\multiseth}_3
    \subseteq\hier{\portsh}_3,
\]
we obtain
\[
\begin{aligned}
\textnormal{co-}\hier{\multiseth}_k
&\subseteq\hier{\multiseth}_{k+3},\\
\textnormal{co-}\hier{\portsh}_k
&\subseteq\hier{\portsh}_{k+3}.
\end{aligned}
\]
The construction is uniform in the port numbering
and also applies to predicates of port-numbered
configurations.

\section{The power of the \seth hierarchy}\label{sec:SETequalsMSO}

In this section, we show that the \seth hierarchy equals MSO. Recall that MSO on graphs uses quantification over vertices, called first-order
quantification, and over sets of vertices, called monadic second-order
quantification. Its atomic formulas are:
\begin{itemize}
    \item \(x_i=x_j\), expressing equality of vertices;
    \item \(x_i\sim x_j\), expressing adjacency;
    \item \(x_i\in X_j\), expressing membership of a vertex in a set.
\end{itemize}

We show first that any graph language expressible in MSO belongs to some level of the \seth hierarchy. Later, we show the opposite direction. 

\subsection{The \seth hierarchy contains MSO}\label{sec:MSOinSET}

We first prove the containment for MSO in the \seth hierarchy.

\lemMSOinSET*

\begin{proof}
Let \(\phi\) be an MSO sentence. We prove the result by induction on
the structure of \(\phi\). Variable assignments are represented by
bit labelings. A set variable \(X\) is represented by an arbitrary
labeling $\ell_X:V\to\{0,1\},$
where \(\ell_X(v)=1\) means \(v\in X\). A first-order variable \(x\)
is represented by a bit labeling
$\ell_x:V\to\{0,1\},$
intended to mark exactly one vertex.

We prove the stronger claim that, for every subformula
\(\psi\), its truth can be verified uniformly for any
fixed assignment of its free variables, represented by
bit labelings. These labelings are treated as previously
assigned certificate fields: they remain visible to the
verifier and its neighbors through the usual \seth view,
and are not quantified again by the subprotocol.

We assume that the labelings of free first-order
variables mark exactly one vertex. Under this
invariant, the atomic formulas are locally checkable.

For \(x=y\), every vertex checks that $\ell_x(v)=\ell_y(v).$
For \(x\sim y\), the unique vertex marked by \(\ell_x\) checks that
the type \(\ell_y=1\) appears among its neighbors. For \(x\in X\),
the unique vertex marked by \(\ell_x\) checks that
\(\ell_X(v)=1\). If the vocabulary contains input predicates
\(P_\sigma\), corresponding to a fixed finite input alphabet, then
\(P_\sigma(x)\) is checked by the unique vertex marked by
\(\ell_x\), which verifies that its own input is \(\sigma\).

Negations of atomic formulas, as well as conjunctions and
disjunctions, can be handled using closure of the \seth hierarchy
under complement, intersection, and union.

It remains to simulate quantification. For a set variable there is
no validity restriction. Thus $\exists X\,\psi$ is simulated by an existentially quantified bit labeling \(\ell_X\),
followed by the protocol for \(\psi\), and similarly $\forall X\,\psi$ is simulated using a universally quantified bit labeling.

For a first-order variable, the labeling must represent a single
vertex. Let \(\operatorname{Sing}(\ell_x)\) denote the property that
\((G,\ell_x)\in\leader\). By \Cref{cor:leader},
\(\operatorname{Sing}\) belongs to the \seth hierarchy.

An existential first-order quantifier is simulated as
$\exists \ell_x\;
    \bigl(
        \operatorname{Sing}(\ell_x)
        \land
        \widehat{\psi}
    \bigr).$
Thus an invalid labeling cannot witness an existential quantifier. A universal first-order quantifier is simulated as $
    \forall \ell_x\;
    \bigl(
        \neg\operatorname{Sing}(\ell_x)
        \lor
        \widehat{\psi}
    \bigr).$ Hence labelings that do not represent a vertex impose no
requirement, while every singleton labeling must satisfy the
translation of \(\psi\). This is exactly the semantics of
universal first-order quantification.

By \Cref{lem:intun,cor:comp}, Boolean combinations
can be implemented uniformly while preserving all
previously assigned certificate fields. Likewise,
each logical quantifier introduces a new bit labeling
before the auxiliary rounds verifying its scope.
Consecutive quantifiers of the same type can be
merged, and dummy rounds can be added when needed.

The induction therefore yields a finite-round
\seth protocol for every subformula, uniformly
in its free-variable assignments. Since \(\phi\)
is a sentence, the resulting protocol has no
free variables and recognizes precisely \(L_\phi\).

Hence every MSO-definable distributed language over a finite input
alphabet belongs to the \seth hierarchy.
\end{proof}

\subsection{MSO contains the \seth hierarchy}\label{sec:SETinMSO}

We now show that the \seth hierarchy captures
exactly MSO-definable languages. Combined with \Cref{lem:MSOinSET}, this gives an exact correspondence between
MSO and the \seth hierarchy.

Let \(A\) be a local verifier in the \seth model
using certificate alphabets
\(\Gamma_1,\ldots,\Gamma_k\), and let \(\Sigma\)
be the finite input alphabet.

The local view of a node consists of its own
input and certificate sequence, together with
the set of certificate sequences appearing
among its neighbors. In particular, a node
does not observe the input values of its
neighbors.

We therefore distinguish two types of local
information. An \emph{own type} is a tuple
\[
\tau=(\sigma,a_1,\ldots,a_k)
\in\Sigma\times\prod_{i\in[k]}\Gamma_i,
\]
while a \emph{neighbor certificate type} is
a tuple
\[
\alpha=(a_1,\ldots,a_k)
\in\prod_{i\in[k]}\Gamma_i.
\]

Each certificate assignment is encoded by
a finite tuple of monadic second-order
variables, one set for each possible
certificate value. These sets are required
to form a partition of the vertex set.

We then construct MSO formulas describing
the own type of a node and the set of
neighbor certificate types it observes.
Since all alphabets are finite, there are
only finitely many possible local views,
and the acceptance condition of the
verifier can be expressed as a finite
disjunction of MSO formulas.

\lemSETinMSO*

\begin{proof}
Let \(L\in \hier{\seth}_k\). By definition, there exists a local verifier
\(A\), finite certificate alphabets \(\Gamma_1,\ldots,\Gamma_k\), and
quantifiers \(\quantif_i\in\{\exists,\forall\}\), for \(i\in[k]\), such that
\((G,x)\in L\) if and only if
\[
\quantif_1\ell_1:V\to\Gamma_1\;
\quantif_2\ell_2:V\to\Gamma_2\;
\cdots\;
\quantif_k\ell_k:V\to\Gamma_k\;
\forall v\in V,\ A(G,x,v,\ell_1,\ldots,\ell_k)=1.
\]

Let \(\Sigma\) be the finite input alphabet. For every input symbol
\(\sigma\in\Sigma\), we use a unary predicate \(P_\sigma\), where
\(P_\sigma(v)\) means that \(x(v)=\sigma\).

For each certificate round \(i\in[k]\), we encode a labeling
\(\ell_i:V\to\Gamma_i\) by a tuple of monadic second-order variables
\( \overline X^i=(X^i_a)_{a\in\Gamma_i}\).
The intended meaning is that \(v\in X^i_a\) if and only if \(\ell_i(v)=a\). We enforce this intended meaning with the partition formula
\[
\Part_i(\overline X^i)
:=
\forall v\,
\left[
\left(\bigvee_{a\in\Gamma_i} v\in X^i_a\right)
\wedge
\bigwedge_{\substack{a,b\in\Gamma_i\\a\neq b}}
\neg\left(v\in X^i_a\wedge v\in X^i_b\right)
\right].
\]
Thus, \(\Part_i(\overline X^i)\) says that every vertex receives exactly one
certificate from \(\Gamma_i\).

\newcommand\translationquantif{T}

We now define the MSO translation of the quantifier prefix. For a formula
\(\psi\), define recursively
\[
\translationquantif_{k+1}(\psi):=\psi,
\]
and, for \(i\in[k]\),
\[
\translationquantif_i(\psi):=
\begin{cases}
\exists \overline X^i\,
\bigl(\Part_i(\overline X^i)\wedge \translationquantif_{i+1}(\psi)\bigr),
&\text{if }\quantif_i=\exists,\\[1ex]
\forall \overline X^i\,
\bigl(\Part_i(\overline X^i)\Rightarrow \translationquantif_{i+1}(\psi)\bigr),
&\text{if }\quantif_i=\forall.
\end{cases}
\]
This ensures that existential and universal quantifiers range only over valid certificate labelings.

It remains to express the local acceptance
predicate of \(A\) in MSO.

We distinguish the finite sets of own types
and neighboring certificate types:
\[
\mathcal T_{\mathrm{own}}
=
\Sigma\times\prod_{i\in[k]}\Gamma_i,
\qquad
\mathcal T_{\mathrm{nbr}}
=
\prod_{i\in[k]}\Gamma_i.
\]

For every own type
\(\tau=(\sigma,a_1,\ldots,a_k)\),
define
\[
\Type_\tau(v)
:=
P_\sigma(v)\wedge
\bigwedge_{i\in[k]}v\in X^i_{a_i}.
\]

For every neighboring certificate type
\(\alpha=(a_1,\ldots,a_k)\), define
\[
\Cert_\alpha(u)
:=
\bigwedge_{i\in[k]}u\in X^i_{a_i}.
\]

Notice that \(\Cert_\alpha(u)\) does not
refer to the input of \(u\), since
neighboring input values are not visible
in the \seth model.

For every \(T\subseteq\mathcal T_{\mathrm{nbr}}\),
define
\[
\begin{aligned}
\SetView_T(v)
:={}&
\bigwedge_{\alpha\in T}
\exists u\,
\bigl(u\sim v\wedge\Cert_\alpha(u)\bigr)
\\
&\wedge
\forall u\,
\left(
u\sim v
\Rightarrow
\bigvee_{\alpha\in T}\Cert_\alpha(u)
\right).
\end{aligned}
\]

The first condition requires every
certificate type in \(T\) to appear
among the neighbors of \(v\), while
the second excludes all other types.
Thus \(\SetView_T(v)\) describes exactly
the neighboring information available
to the \seth verifier.

Since \(A\) is fixed and all alphabets
are finite, its accepting local views
form a finite set $
\mathcal A
\subseteq
\mathcal T_{\mathrm{own}}
\times 2^{\mathcal T_{\mathrm{nbr}}}.
$ We define $
\Accept_A(v)
:=
\bigvee_{(\tau,T)\in\mathcal A}
\left(
\Type_\tau(v)\wedge\SetView_T(v)
\right).
$

By construction, \(\Accept_A(v)\) holds
if and only if the original \seth
verifier \(A\) accepts at \(v\).

Finally, define the MSO sentence
$
\varphi_L := \translationquantif_1\bigl(\forall v\,\Accept_A(v)\bigr).$
By the construction above, for every input graph \((G,x)\),
\[
(G,x)\in L
\quad\text{if and only if}\quad
(G,x)\models\varphi_L.
\]
Therefore \(L\) is MSO-definable.

We can also bound the quantifier depth of the resulting
formula in terms of the certificate alphabets. The formula
\(\forall v\,\Accept_A(v)\) has quantifier depth at most \(2\).
For each round \(i\), the tuple \(\overline X^i\) contains
\(|\Gamma_i|\) monadic second-order variables, while the
partition condition \(\Part_i(\overline X^i)\) has quantifier
depth \(1\). Since this condition is combined with the
remaining translation by conjunction or implication, the
recursive definition of \(\translationquantif_i\) gives
a quantifier depth of at most
\[
2+\sum_{i=1}^{k}|\Gamma_i|
\]
for \(\varphi_L\). Thus the bound depends on the certificate
alphabets of the chosen protocol, as well as on the number
of certificate rounds.
\end{proof}

\subsection{Trading horizon for interaction}
\label{sec:horizon-set}

In this subsection, we discuss extensions of the \seth model in which the
verifier can inspect certificates beyond its immediate neighborhood. We first
explain that there are two natural ways to define such a model, denoted the \emph{weak \(r\)-horizon \seth model} and the \emph{strong \(r\)-horizon \seth
model}.  The weak
model need not become more expressive as the
horizon increases. For the strong model, we
show that any fixed horizon can be simulated
at horizon one with at most one additional
interaction.

In the \emph{weak \(r\)-horizon \seth model}, a node
\(v\) always has access to its own input \(x(v)\)
and its own certificate sequence
\(\overline{\ell}(v)\), separately from the
certificates of other vertices. In addition, it
observes the set
$W_r(v)=
\left\{
\overline{\ell}(u):
u\in B_r(v)\setminus\{v\}
\right\},$
where
$B_r(v)=\{u\in V:\dist(u,v)\leq r\}.$
Thus, the node observes the set of certificate
sequences appearing at other vertices within
distance at most \(r\), without distinguishing
their distances or multiplicities.

In particular, for \(r=1\),
$W_1(v)=
\{\overline{\ell}(u):u\in N(v)\},$
so the weak \(1\)-horizon model coincides exactly
with the original \seth model.

The power of the weak model is not monotone in the horizon. In particular, the weak
\(2\)-horizon model does not necessarily simulate the usual radius-one model.
Consider the language
\[
\textsc{Star}=\{G: G\simeq K_{1,n}\text{ for some }n\geq 2\}.
\]

This language is recognizable in the weak
\(1\)-horizon \seth model with two alternating
certificate assignments.

In the first round, the prover assigns each
vertex a role in \(\{C,L_0,L_1\}\), where \(C\)
denotes a center and \(L_0,L_1\) denote two
types of leaves. In the second round, the
disprover assigns an arbitrary bit \(d(v)\)
to each vertex.

Each vertex performs the following checks
using its own certificates and the set of
certificate sequences of its neighbors.

A vertex marked \(C\) checks that it has no
neighbor marked \(C\), and that it has at
least one neighbor of type \(L_0\) and
at least one of type \(L_1\).

A vertex marked \(L_0\) or \(L_1\) checks
that all its neighbors are marked \(C\),
and that the set
\[
\{d(u):u\in N(v),\ \sigma(u)=C\}
\]
contains exactly one element, where
\(\sigma\) denotes the first-round role
assignment.

If \(G=K_{1,n}\) with \(n\geq2\), the prover
marks the center with \(C\), one leaf with
\(L_0\), and all remaining leaves with \(L_1\).
Every verification condition then holds,
regardless of the disprover's bits.

Conversely, suppose that a first-round role
assignment is accepted for every assignment
of the disprover's bits. There can be no
edge between two centers or between two
leaves, by the local role checks. Moreover,
every leaf has exactly one center neighbor:
it must have at least one, and if it had
two distinct center neighbors, the disprover
could assign different bits to them, causing
the leaf to reject.

Since \(G\) is connected, these conditions
imply that there is exactly one center and
that every other vertex is adjacent only
to this center. Finally, the center must
have at least two leaves, since it sees
both types \(L_0\) and \(L_1\). Therefore,
\(G\simeq K_{1,n}\) for some \(n\geq2\).

However, \(\textsc{Star}\) is not recognizable in the weak
\(2\)-horizon model. Consider, for every \(n\geq 2\), the graphs \(K_{1,n}\)
and \(K_{n+1}\) on the same vertex set. Both graphs have diameter at most \(2\).
Hence, in the weak \(2\)-horizon model, for every node \(v\) in either graph,
the ball \(B_2(v)\) is the whole vertex set. It follows that no weak \(2\)-horizon \seth protocol can distinguish
\(K_{1,n}\) from \(K_{n+1}\). 
Thus \(\textsc{Star}\) is not recognizable in the weak \(2\)-horizon \seth
hierarchy.

To obtain a monotone extension of the original \seth
model, we introduce the following definition.

Let \(r\geq1\). In the \emph{strong \(r\)-horizon
\seth model}, suppose that a protocol uses \(k\)
certificate assignments $\ell_i:V\to\Gamma_i$ for  $i\in[k].$
Write $
\gamma(v)=(\ell_1(v),\ldots,\ell_k(v))$
for the certificate sequence of vertex \(v\).

After the certificate assignments have been made,
each vertex \(v\) has access to its own input
\(x(v)\) and its own certificate sequence
\(\gamma(v)\), separately from the certificates
of other vertices.

In addition, the vertex observes the tuple
$
\left(
S_N(v),S_{\leq2}(v),\ldots,S_{\leq r}(v)
\right),
$
where
$
S_N(v)=\{\gamma(u):u\in N(v)\}
$
is the usual \seth neighborhood view, and, for
every \(2\leq j\leq r\),
$S_{\leq j}(v)=
\{\gamma(u):u\in B_j(v)\}.$

The sets \(S_{\leq j}(v)\) include the
certificate sequence of \(v\), since
\(v\in B_j(v)\). However, the immediate
neighborhood view \(S_N(v)\) is provided
separately. Thus, a vertex can always
distinguish its own certificates from the
presence of the same certificates among
its immediate neighbors.

As usual in the \seth model, multiplicities
are not visible in any of these sets. When \(r=1\), the tuple consists only of
\(S_N(v)\). Hence the strong \(1\)-horizon
model coincides exactly with the original
\seth model. We denote the corresponding
\(k\)-th level by \(\sethr{r}{k}\), so that $\sethr{1}{k}=\hier{\seth}_k.$ Moreover, for every \(r\geq1\) and every
\(k\geq0\),
$
\sethr{1}{k}\subseteq\sethr{r}{k},
$
since a strong \(r\)-horizon verifier can
ignore the additional ball information
and use only \(S_N(v)\).

We now prove the converse inclusion with
at most one additional interaction.

\thmSEThorizon*

\begin{proof}
Let \(L\in\sethr{r}{k}\), and let \(A\) be a
strong \(r\)-horizon \seth verifier recognizing
\(L\). We construct a radius-one \seth protocol
using at most one additional certificate round.

Let \(\ell_i:V\to\Gamma_i\), for \(i\in[k]\),
be the certificate assignments of the original
protocol. Write
\(\mathcal C=\Gamma_1\times\cdots\times\Gamma_k\)
for the finite alphabet of complete certificate
sequences, and let
\(\gamma(v)=(\ell_1(v),\ldots,\ell_k(v))
\in\mathcal C\).
For \(k=0\), we take \(\mathcal C\) to consist
of the unique empty sequence.

The prover supplies to every vertex \(v\)
a tuple
\((T_0(v),T_1(v),\ldots,T_r(v))\),
where each \(T_j(v)\subseteq\mathcal C\).
These sets are intended to satisfy
\(T_j(v)=\{\gamma(u):u\in B_j(v)\}\)
for \(0\leq j\leq r\).

The tuple is supplied in an existential
certificate assignment after all original
certificate fields have been determined.
If the original protocol ends with an
existential assignment, these additional
fields are merged into that assignment.
Otherwise, we append one existential round.

We claim that the intended values of the
sets \(T_j(v)\) can be verified using only
radius-one \seth information.

First, every vertex \(v\) checks
\(T_0(v)=\{\gamma(v)\}\).
Then, for every \(j\in[r]\), it checks
\(T_j(v)=T_0(v)\cup
\bigcup_{u\in N(v)}T_{j-1}(u)\).
This condition is locally checkable in
the \seth model: the vertex reads its own
set \(T_0(v)\) and the set of values
\(T_{j-1}(u)\) appearing among its neighbors.
Multiplicities are irrelevant because the
operation is set union.

We now show that if all vertices pass
these checks, every claimed set is correct.
The assertion holds for \(j=0\) by the
first check. Suppose it holds for \(j-1\).
Since
\(B_j(v)=\{v\}\cup
\bigcup_{u\in N(v)}B_{j-1}(u)\),
we obtain
\[
\begin{aligned}
T_j(v)
&=
\{\gamma(v)\}\cup
\bigcup_{u\in N(v)}T_{j-1}(u)\\
&=
\{\gamma(u):u\in B_j(v)\}.
\end{aligned}
\]
Thus the assertion follows by induction
for every \(0\leq j\leq r\).

Once these sets have been verified,
the radius-one verifier reconstructs
the strong \(r\)-horizon view required
by \(A\).

The vertex already knows its own input
\(x(v)\) and original certificate sequence
\(\gamma(v)\). Moreover, it directly observes
the original \seth neighborhood view
\(S_N(v)=\{\gamma(u):u\in N(v)\}\)
by projecting the certificate sequences
of its neighbors onto the fields of the
original \(k\) rounds.

For each \(2\leq j\leq r\), the remaining
component of the strong view is
\(S_{\leq j}(v)=T_j(v)\).
Therefore the verifier can reconstruct
exactly the tuple
\((S_N(v),S_{\leq2}(v),\ldots,S_{\leq r}(v))\)
and simulate the original verifier \(A\).
If any consistency check fails, the new
verifier rejects.

We now establish correctness. Fix any
assignment of the original \(k\)
certificate fields.

If the original verifier accepts at
every vertex, the prover can supply
the correct sets
\(T_j(v)=\{\gamma(u):u\in B_j(v)\}\).
Every consistency check succeeds, and
the new verifier reconstructs the exact
strong \(r\)-horizon view at every vertex.
Hence all vertices accept.

Conversely, suppose that some assignment
of the additional sets makes all vertices
accept. The consistency checks force all
the sets \(T_j(v)\) to have their intended
values. Consequently, the new verifier
simulates exactly the strong \(r\)-horizon
verifier \(A\), which must therefore
accept at every vertex.

Thus, for every fixed assignment of
the original certificate fields, there
exists an accepting assignment of the
additional sets if and only if the
original verifier accepts everywhere.
This equivalence preserves the entire
alternating quantifier prefix of the
original protocol.

Finally, all certificates have constant
size: the alphabet \(\mathcal C\) is finite,
\(r\) is fixed, and each additional
certificate consists of \(r+1\) subsets
of \(\mathcal C\).

If the original protocol ends with a
universal assignment, the new existential
round extends its alternating prefix by
one round. If it ends with an existential
assignment, the additional fields are
merged into that round, so no extra round
is needed. For \(k=0\), one existential
round suffices. Whenever necessary, a
dummy round can be appended to obtain
exactly \(k+1\) alternating rounds.

Therefore,
\(\sethr{r}{k}\subseteq\sethr{1}{k+1}\).
Together with
\(\sethr{1}{k}\subseteq\sethr{r}{k}\),
this implies
\(\bigcup_{k\geq0}\sethr{r}{k}
=\bigcup_{k\geq0}\sethr{1}{k}
=\hier{\seth}\).
\end{proof}

\subsection{Comparison with distributed graph automata}

In \cite{Reiter_2015}, an equivalence is established between languages recognized by MSO formulae and those recognized by alternating distributed graph automata.

Alternating distributed graph automata are automata using a finite set of states $Q$ with a local transition of the form\footnote{This  is the simplified form from \cite{Reiter_2015} when restricting to graph without edge labeling, to fit the formalism of the present paper.}: \(Q\times 2^Q\to 2^Q\). This local transition map is applied uniformly at each node of the input graph on the state of the node together with the set of states of neighbors. Together with the initialization map used to read input labels on nodes locally, it implements the same mode of communication as in the \seth model.

Moreover, states are leveled and typed level-wise as existential, universal or permanent. The transition map can only increase the level until reaching a permanent state, so that there is convergence after a bounded number of applications. Universal and existential states have the same interpretation as in classical alternating automata. Since the type of states is the same at each level, this is equivalent to alternating a bounded number of prover/disprover rounds and letting a local \seth algorithm decide whether it (locally) accepts or not given all the local labelings introduced by the prover/disprover interaction.

The key difference between these two models lies in the accepting condition: for alternating distributed graph automata, a global condition on the set of permanent states present in the graph after convergence is used, whereas in the \seth model the unanimity condition is used (all nodes must accept locally) as in the LOCAL model and other hierarchies based on it \cite{reiter2024local,feuilloley2021hierarchy}. In the formalism of \cite{Reiter_2015}, the unanimity condition is an accepting condition consisting of one ``accepting set'' which is the singleton \(\{\mathtt{accept}\}\) (meaning that a run is accepting if and only if the set of permanent sets in the graph once stabilized is exactly \(\{\mathtt{accept}\}\)).

Thus, the \seth model can be seen as a particular case of alternating distributed graph automata, and Theorem~\ref{thm:set-equals-mso} can be seen as a strengthened version of the equivalence result of \cite{Reiter_2015} using a simpler accepting condition, which is natural in the context of distributed computing. This difference in accepting condition is not just a technical convention, and we would like to point out three noticeable differences between the two models.

\subparagraph*{Levels and number of alternations.} The two models involve quantifier alternation (or levels) but they are not equivalent level-wise. In particular, the ALOS problem which is key in establishing Theorem~\ref{thm:set-equals-mso} is at level three of the \({\seth}\) hierarchy (Theorem~\ref{th:alos}), but it is trivial in the formalism of alternating distributed graph automata: just take the initialization map that read the local input label into states and use as accepting condition the two sets \(\{1\}\) and \(\{1,0\}\) (thus requiring at least one $1$ among input labels on nodes). Thus ALOS doesn't require any kind of alternation in the formalism of \cite{Reiter_2015}, which places it at the same level as purely LOCAL languages.

\subparagraph{Non-connected graphs.} If one extends the \seth model to non-connected graphs, the accepting condition is such that a graph is accepted if and only if all of its connected components are accepted. In particular ALOS on non-connected graphs cannot be recognized in the \seth model. With the global accepting condition of \cite{Reiter_2015}, there is no such limitation and ALOS is recognized on non-connected graphs using the same distributed graph automaton as above.

Deciding if a graph is connected (still a MSO property) is easy with distributed graph automata. Indeed, with one level of universal states one can put a random bit on each node and goes to permanent state \(b\in\{0,1\}\) if $b$ is the only random bit seen in the neighborhood and to the $\mathtt{reject}$ permanent state otherwise. Then the accepting condition is made of sets \(\{1\}\), \(\{0\}\) and all possible sets containing \texttt{reject}. This characterizes connectedness because non-connected graphs can lead to the set of permanent states \(\{0,1\}\), and a connected graph will never produce \(\{0,1\}\) as final set of permanent states because if both bits $0$ and $1$ are present in the graph, then a node will see them both in its neighborhood and turn into \texttt{reject}. Connectedness is however impossible to capture with model \seth.

\subparagraph{Infinite graphs.} In both models, no identifier are used on nodes to distinguish them, so they both make perfect sense on arbitrary infinite graphs. Although \cite{Reiter_2015} only deals with finite graphs, it seems clear that the equivalence with MSO holds on infinite graphs (in particular one can check that the automaton for ALOS above perfectly works on infinite graphs). We would like to stress that the equivalence between \seth and MSO breaks on infinite graphs for a deep reason: recognized language with the unanimity acceptance condition are topologically closed, but some MSO languages like ALOS are not. Although it goes beyond the scope of the present paper, let us make this statement more precise.

Fix $G=\Z$, the bi-infinite line graph ($i\in\Z$ is connected to $i+1$ and $i-1$), and let us restrict to languages made only of labelings of $G$, or equivalently, languages of bi-infinite words. For a finite alphabet \(\Sigma\), the space
\(\Sigma^{\mathbb Z}\), equipped with the product
of the discrete topology on \(\Sigma\), is compact. The language ALOS restricted to $G$ is not a closed set in this topology: indeed, the sequence of yes-instances where the unique occurrence of $1$ is at position $i$ converges, when \(i\to\infty\), to the all $0$ configuration which is a no-instance. However, we claim that any language of bi-infinite words recognized in the \seth hierarchy is closed, and hence that ALOS is not recognizable in \seth for infinite graphs.

This closedness property roughly holds for the following reasons: first, acceptance at a given node $i$ is a closed condition because it depends on finitely many positions (neighbors of $i$). So unanimous acceptance is an intersection of closed set, hence a closed set. This is true even in the context of a fixed free variable not part of the input labeling. Second, the property is preserved when climbing up in the hierarchy because quantification has the following topological interpretation: on one hand, universal quantification corresponds to intersection, and, on the other hand, existential quantification corresponds to a continuous projection (hence sending closed sets to closed sets by the closed map lemma since the ambient space is compact). As an alternative proof, one can check that \seth-recognized languages belongs to the u-MSO fragment of \cite{T_rm__2014} and then use Corollary 1 therein.

This topological constraint disappears on finite graphs (which do not form a compact set), and the finiteness property is crucially used in lemmas~\ref{lem:dycycle} (all nodes having a successor implies existence of a cycle) and~\ref{lemm:XORsurj} (injectivity is equivalent to surjectivity for finite maps). Both lemmas are false on infinite graphs.

\section{Counting in \portsh and comparison with CMSO}
\label{sec:MSOCMSOinH}

We now turn to CMSO. Recall that CMSO extends MSO with counting predicates
\(\cardmod{q}{p}{X_i}\), which express cardinality modulo a fixed positive
integer. Here we prove their verification in \portsh; the weaker \multiseth
model is considered separately below.

\subsection{Counting modulo a constant}
\label{sec:counting}

The key point is not merely to count neighboring certificates but to obtain a
rooted tree on \emph{every} connected graph with constant-size vertex labels.
Demanding that each non-root vertex have exactly one neighbor of the preceding
color modulo three does not suffice: such a coloring need not exist (already on
\(K_{3,3}\)). Instead, in \portsh every node selects the first eligible
neighbor according to its fixed local ordering of ports.

Fix an arbitrary port numbering \(\pi\) of \(G=(V,E)\), a root marking
\(r:V\to\{0,1\}\), and a coloring \(t:V\to\{0,1,2\}\). Define
\[
S_t(v)=\{u\in N(v):t(u)=t(v)-1\pmod 3\}.
\]
If \(S_t(v)\ne\emptyset\), let \(\parent_\pi(v)\) be the neighbor in
\(S_t(v)\) incident to the smallest local port of \(v\). This selects a
\emph{single} neighbor regardless of how many eligible neighbors share the
same labels. At a root, or if no eligible neighbor exists, set \(\parent_\pi(v)=\bot\).

Let \(\portspann\) denote the property of a port-numbered configuration
\((G^\pi,(r,t))\) satisfying:
\begin{enumerate}
\item exactly one node \(r_0\) has \(r(r_0)=1\);
\item \(S_t(r_0)=\emptyset\), and \(S_t(v)\ne\emptyset\) for every
      \(v\ne r_0\);
\item the functional digraph \(T_{r,t,\pi}\), with edges
      \(v\to\parent_\pi(v)\) for \(v\ne r_0\), is acyclic.
\end{enumerate}
The selected parent relation, and hence \(\portspann\),
may depend on the port numbering \(\pi\). We use
\(\portspann\) only as an auxiliary condition in the
counting protocol, which must recognize the same language
for every port numbering.

We say that such a predicate admits a uniform \(k\)-round
\portsh protocol if the same verifier and certificate
alphabets recognize it for every fixed port numbering,
using \(k\) alternating certificate assignments.

Under these conditions \(T_{r,t,\pi}\) is a spanning tree directed toward
\(r_0\): in a finite acyclic functional digraph every path reaches a vertex
without a parent, and this vertex must be the unique root.

\begin{corollary}\label{cor:span-ports}
The predicate \(\portspann\) admits a uniform
three-round \portsh protocol with quantifier pattern
\(\exists\forall\exists\).

Moreover, for every connected graph \(G\) and every
port numbering \(\pi\), there exist \(r,t\) satisfying
\(\portspann\).
\end{corollary}

\begin{proof}
The unique-root requirement is the language \(\leader\), which belongs to
\(\hier{\seth}_3\subseteq\hier{\portsh}_3\).
The first existential round can also copy the bounded input \(r,t\)
into certificates and locally check that the copies are faithful, as in
\Cref{rem:simple}. Each vertex can then inspect the copies at its ports,
verify the second condition, and identify its selected parent. The universal
round supplies \(y\), and the final existential round supplies \(z\).

For acyclicity, use \Cref{lemm:XORsurj} on \(T_{r,t,\pi}\). For a universal
bit labeling \(y\), the final prover supplies a bit labeling \(z\). A root
checks \(y(v)=z(v)\), while every other node checks
\[
y(v)=z(v)+z(\parent_\pi(v))\pmod 2.
\]
Each equality is local in \portsh: a node reads \(z\) at the first eligible
port. The condition \(\forall y\,\exists z\) holds exactly when the selected
functional digraph is acyclic. Hence this condition lies at level two.
Combining these checks with the three-round
protocol for \(\leader\) gives a uniform
\(\exists\forall\exists\) protocol for \(\portspann\).
All checks use the same fixed port numbering.

Finally, let \(r_0\) be any node and put \(t(v)=\dist_G(r_0,v)\bmod 3\),
with \(r\) marking \(r_0\). Every non-root vertex has a neighbor one step
closer to the root, while the root has no preceding-color neighbor.
Along every eligible edge \(v\to u\), BFS distance strictly decreases:
adjacent vertices have distances differing by at most one, and the
preceding-color condition forces the difference to be exactly one.
Consequently every port-selected parent relation is acyclic, for every
\(\pi\). This proves the existence assertion.
\end{proof}

For \(p\ge1\) and \(0\le q<p\), define
\[
\lcount{q,p}=\{(G,x):x\in\{0,1\}^V,\quad
\lvert M(x)\rvert\equiv q\pmod p\},
\qquad M(x)=\{v\in V:x(v)=1\}.
\]
Any other integer representative \(q\) is handled by replacing it with its
residue modulo \(p\).

As usual, a fixed finite alphabet can be copied into a constant-size
certificate in an existential round, making that information visible to
neighbors; cf. the following observation.

\begin{remark}\label{rem:simple}
When the input alphabet is bounded, one existential quantifier can make
bounded input information locally visible to the verifier. For example,
if \(x\in\{0,1\}^k\) for fixed \(k\), then
\[
\forall v\,\forall u\in N(v),\quad x(v)=x(u)
\]
holds if and only if there exists a labeling \(\ell\) such that each node
checks \(\ell(v)=x(v)\) and \(\ell(u)=\ell(v)\) at all its neighbors.
This works in each of the three local models.
\end{remark}

\thmCount*

\begin{proof}
Fix an arbitrary port numbering \(\pi\). The initial prover supplies
\(\ell_r:V\to\{0,1\}\) and \(\ell_t:V\to\{0,1,2\}\), intended to certify
\(\portspann\). By the preceding corollary, a suitable pair exists for
all connected yes-instances, independently of \(\pi\).

The prover also supplies \(\ell_c:V\to\{0,\ldots,p-1\}\). For each
vertex \(v\), the intended recurrence is
\begin{equation}\label{eq:computree}
\ell_c(v)\equiv x(v)+
\sum_{u:\parent_\pi(u)=v}\ell_c(u)\pmod p.
\end{equation}
Here the sum ranges over the \emph{children} of \(v\) in the port-selected
tree, not over all its neighbors. The prover must also satisfy the local
condition \(\ell_r(v)=1\Rightarrow\ell_c(v)=q\). If all recurrences hold,
induction from the leaves gives
\(\ell_c(r_0)\equiv \lvert M(x)\rvert\pmod p\).

To test all these equations without allowing a vertex to inspect the
\emph{port choices of its neighbors}, the disprover marks a challenge vertex
by a bit labeling \(\ell_p:V\to\{0,1\}\). When
\((G,\ell_p)\in\leader\), let \(v^*\) be the uniquely marked node. If the
challenge is not in \(\leader\), it is ignored.

The final prover supplies \(\ell:V\to\{0,\ldots,p-1\}\). We require
\begin{equation}\label{eq:erase}
\ell(u)=
\begin{cases}
\ell_c(u), &\text{if }u\ne r_0\text{ and }
\ell_p(\parent_\pi(u))=1,\\
0, &\text{otherwise.}
\end{cases}
\end{equation}
This equation is locally checkable in \portsh: a non-root node selects its
parent using the first eligible port, reads the parent's challenge bit, and
checks its own value \(\ell(u)\). The root checks \(\ell(r_0)=0\).

At any challenged vertex \(v\) the verifier also checks
\begin{equation}\label{eq:local-sum}
\ell_c(v)\equiv x(v)+\sum_{u\in N(v)}\ell(u)\pmod p.
\end{equation}
This sum is local in \portsh, which presents one certificate value per
incident port (including repeated values). If the challenge is unique and
\eqref{eq:erase} holds, the sum in \eqref{eq:local-sum} is exactly the sum
of \(\ell_c\) over the children of \(v^*\). Thus
\eqref{eq:local-sum} tests \eqref{eq:computree} at \(v^*\).

We now combine these conditions. Fix a port numbering
\(\pi\). The protocol uses the outer quantifier prefix
\[
    \exists(\ell_r,\ell_t,\ell_c)
    \quad\forall\ell_p
    \quad\exists\ell.
\]
After these assignments, the remaining verification
requires:
\begin{enumerate}
    \item
    \((G^\pi,(\ell_r,\ell_t))\in\portspann\).

\item
Every vertex \(v\) checks that \(x(v)\in\{0,1\}\)
and that \(\ell_r(v)=1\Rightarrow\ell_c(v)=q\).

    \item
    Either \((G,\ell_p)\notin\leader\), or
    \eqref{eq:erase} holds at every vertex and
    \eqref{eq:local-sum} holds at the unique vertex
    marked by \(\ell_p\).
\end{enumerate}

The first condition is verified by the uniform
protocol of \Cref{cor:span-ports}. The second is
directly local. In the third condition, the two
equations are also directly checkable in \portsh:
all the labelings involved have already been chosen
in the outer rounds and are therefore available
as neighboring certificate fields.

Although the first and third conditions may depend
on \(\pi\), the complete protocol uses the same
verification algorithms for every port numbering.
For every fixed \(\pi\), these conditions are
equivalent to membership in \(\lcount{q,p}\).
Here the outer existential labels have constant size because \(p\) is fixed.
For completeness, use a BFS-distance coloring and the true subtree sums;
for each valid challenge the final prover uses the uniquely determined
values in \eqref{eq:erase}. 
For soundness, invalid inputs are rejected locally,
and invalid tree certificates are caught by the
\(\portspann\) check. For a valid binary input
and a valid tree, an incorrect root value causes
rejection; otherwise, the disprover challenges
a vertex where the recurrence fails.
 Equation \eqref{eq:erase} forces exactly its children's
contributions, and \eqref{eq:local-sum} rejects for every final reply.


We now make the number of rounds explicit.
By \Cref{cor:span-ports}, \(\portspann\) has an
\(\exists\forall\exists\) protocol. Since
\(\leader\) has the same quantifier pattern,
\Cref{cor:comp} gives a
\(\forall\exists\forall\exists\forall\exists\)
protocol for \(\overline{\leader}\).
Denote their certificate assignments by
\(\alpha_1,\alpha_2,\alpha_3\) and
\(\beta_1,\ldots,\beta_6\), respectively.

Let \(\mathcal E\) denote the local condition
requiring \eqref{eq:erase} everywhere and
\eqref{eq:local-sum} at every challenged vertex.
Condition (3) above is precisely
\(\overline{\leader}(\ell_p)\lor\mathcal E\).
We implement this disjunction using a globally
constant existential selector \(b\), chosen
before the universal challenges of the
\(\overline{\leader}\) protocol.

The complete alternating prefix is
\[
\begin{aligned}
&\exists(\ell_r,\ell_t,\ell_c)\;
\forall\ell_p\;
\exists(\ell,b,\alpha_1)\;
\forall(\alpha_2,\beta_1)\\
&\exists(\alpha_3,\beta_2)\;
\forall\beta_3\;
\exists\beta_4\;
\forall\beta_5\;
\exists\beta_6.
\end{aligned}
\]
The \(\alpha\)-fields verify \(\portspann\).
If \(b=0\), the \(\beta\)-fields verify
\(\overline{\leader}(\ell_p)\); if \(b=1\),
the verifier checks \(\mathcal E\) directly.
The root condition is checked locally.

Correctness follows from the preceding argument:
on yes-instances, the prover certifies a valid
tree and correct subtree sums, choosing the
appropriate branch for each challenge. On
no-instances, either the tree or root check
fails, or the disprover selects a vertex
with an incorrect recurrence. In the latter
case, neither branch can succeed.

All certificate alphabets are finite, and
the construction is uniform in the port
numbering. Thus
\[
\lcount{q,p}\in\hier{\portsh}_9.
\]

\end{proof}

This bound is an upper bound only. In particular, we do not assert that the
same protocol works in \multiseth: an unnumbered multiset does not let a
node choose the first eligible port when several neighbors carry the same
certificate.

\subsection{The \portsh hierarchy contains CMSO}
\label{sec:CMSOinPORTS}

\thmCMSO*

\begin{proof}
Let \(\phi\) be a CMSO sentence. We use the same inductive
translation as in the proof of \Cref{lem:MSOinSET}.

First-order variables are represented by singleton bit labelings,
and set variables by arbitrary bit labelings. The non-counting
atomic formulas $x=y, x\sim y, x\in X$
are handled exactly as in the \seth case, since a \portsh verifier
can simulate a \seth verifier.

For a counting atom $\cardmod{q}{p}{X},$
the bit labeling representing \(X\) defines an instance of
\(\lcount{q',p}\), where \(q'=q\bmod p\). By
\Cref{lemm:count}, this language belongs to the \portsh hierarchy.
Negated atoms and Boolean combinations are handled using closure
under complement, finite intersection, and finite union.

The logical quantifiers are treated as in
\Cref{lem:MSOinSET}. Set quantifiers range over arbitrary bit
labelings. For a first-order variable \(x\), let
\(\operatorname{Sing}(\ell_x)\) denote the property that
\(\ell_x\) marks exactly one vertex. Then
$\exists x\,\psi$
is simulated by $
    \exists \ell_x\;
    \bigl(
        \operatorname{Sing}(\ell_x)
        \land
        \widehat{\psi}
    \bigr), $
whereas $\forall x\,\psi$
is simulated by $
    \forall \ell_x\;
    \bigl(
        \neg\operatorname{Sing}(\ell_x)
        \lor
        \widehat{\psi}
    \bigr).$
Since \(\leader\in\hier{\seth}\subseteq\hier{\portsh}\),
these guards are available in the \portsh hierarchy.

Applying the construction recursively yields a constant-size
alternating certification protocol with finitely many rounds.
Therefore every CMSO-definable language over a finite input
alphabet belongs to \(\hier{\portsh}\).
\end{proof}

\section{The \multiseth hierarchy is not included in CMSO}

In this section we exhibit a problem in the \multiseth hierarchy that is not expressible in CMSO, thereby proving that this hierarchy is not contained in CMSO, and also showing that the \portsh hierarchy is strictly more expressive than CMSO.

\thmExout*

\begin{proof}
  Consider the language $L$ of graphs (without input labels) such that each node of degree $3$ or more has an equal number of neighbors of degree $1$ and of degree $2$.
  It is clear that $L \in \hier{\multiseth}_1$ because each node can make publicly available whether its degree is \(1\) or \(2\) through an existential certificate (see \Cref{rem:simple}); then each node can check locally the desired property on its local view.
  
  We now show that \(L\) is not CMSO-definable.
Suppose, for a contradiction, that there exists a CMSO
sentence \(\phi\) defining \(L\).

Consider words of the form
\(w=\#a^m(bc)^k\), where \(m,k\geq2\).
We associate with each such word a graph \(G_w\) whose
vertices are the positions of \(w\). The first position,
labeled \(\#\), is adjacent to every position labeled
\(a\) or \(b\). In addition, each position labeled \(b\)
is adjacent to the position labeled \(c\) immediately
following it. There are no other edges.

The resulting graph is simple and connected. Its first
vertex \(r\) has degree \(m+k\). Every vertex labeled
\(a\) or \(c\) has degree one, while every vertex labeled
\(b\) has degree two. Thus \(r\) is the only vertex of
degree at least three, and it has exactly \(m\) neighbors
of degree one and \(k\) neighbors of degree two.
Consequently, \(G_w\in L\) if and only if \(m=k\).

The graph \(G_w\) is first-order interpretable in the
word \(w\). Indeed, its vertices are the positions of
the word, and its adjacency relation is defined by
\[
\begin{aligned}
E(x,y)\equiv{}&
  \bigl(P_\#(x)\land
        (P_a(y)\lor P_b(y))\bigr)\\
&\lor
  \bigl(P_\#(y)\land
        (P_a(x)\lor P_b(x))\bigr)\\
&\lor
  \bigl(P_b(x)\land P_c(y)
        \land\operatorname{succ}(x,y)\bigr)\\
&\lor
  \bigl(P_b(y)\land P_c(x)
        \land\operatorname{succ}(y,x)\bigr),
\end{aligned}
\]
where \(P_\sigma(x)\) means that position \(x\) carries
symbol \(\sigma\), and \(\operatorname{succ}(x,y)\)
means that \(y\) immediately follows \(x\).

Replacing each adjacency atom in \(\phi\) by this
formula yields a CMSO sentence \(\phi'\) over words,
such that \(w\models\phi'\) if and only if
\(G_w\models\phi\)
for every word of the form \(\#a^m(bc)^k\) with
\(m,k\geq2\).

Since the language
\(W=\{\#a^m(bc)^k:m,k\geq2\}\)
is regular, it is MSO-definable. Therefore, intersecting
\(\phi'\) with a formula defining \(W\) gives a CMSO
sentence defining
\(\{\#a^k(bc)^k:k\geq2\}\).
However, this language is not regular: the prefixes
\(\#a^k\), for \(k\geq2\), are pairwise distinguishable
by the suffixes \((bc)^k\).

This contradicts the fact that every CMSO-definable
language of finite words is regular
(see \cite[Theorem 1.16 and Section 5.2.7]{courcellebook}).
Hence \(L\) is not CMSO-definable.
\end{proof}

\section{Separating the \multiseth and \portsh hierarchies}
\label{sec:multiset-ports}

We now show that port numbers strictly increase the expressive power of
constant-size alternating certification. The separation already holds for
graph languages without input labels.

\thmMultisetPorts*

The inclusion is immediate: a \portsh verifier can discard port numbers
and recover the multiset of neighboring certificate sequences. To prove
strictness, we encode binary words by graphs. We show that \portsh can
recognize every decidable word language on this encoding, provided that
it is closed under reversal. We then construct such a language that no
\multiseth protocol recognizes.

\subsection{Encoding words by graphs.}
For a word \(w=w_1\cdots w_n\in\{0,1\}^n\), with \(n\geq6\), let
\(G_w\) be obtained from the path \(v_1,\ldots,v_n\) by adding a vertex
\(r\) adjacent to every path vertex, and by attaching one leaf to
\(v_i\) if and only if \(w_i=1\). \Cref{fig:word-encoding} illustrates this construction for
\(w=1011010\).

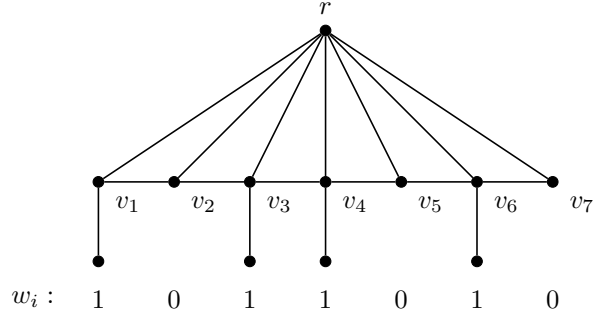
\begin{figure}[t]
\centering
\begin{tikzpicture}[
    x=1cm,
    y=1cm,
    font=\small,
    vertex/.style={
        circle,
        draw,
        fill=black,
        inner sep=0pt,
        minimum size=4pt
    },
    edge/.style={line width=0.6pt}
]
    \foreach \i in {1,...,7} {
        \node[
            vertex,
            label={[xshift=1pt]below right:\(v_{\i}\)}
        ] (v\i) at ({\i-1},0) {};
    }

    \draw[edge]
        (v1)--(v2)--(v3)--(v4)--(v5)--(v6)--(v7);

    \node[vertex, label={above:\(r\)}]
        (r) at (3,2) {};

    \foreach \i in {1,...,7} {
        \draw[edge] (r)--(v\i);
    }

    \foreach \i in {1,3,4,6} {
        \node[vertex] (leaf\i) at ({\i-1},-1.05) {};
        \draw[edge] (v\i)--(leaf\i);
    }

    \node[anchor=east] at (-0.45,-1.55) {\(w_i:\)};

    \foreach \i/\bit in {
        1/1, 2/0, 3/1, 4/1, 5/0, 6/1, 7/0
    } {
        \node at ({\i-1},-1.55) {\(\bit\)};
    }
\end{tikzpicture}
\caption{The graph \(G_w\) for \(w=1011010\).
The vertices \(v_1,\ldots,v_7\) induce a path, and the hub \(r\)
is adjacent to every path vertex. A pendant leaf is attached to
\(v_i\) exactly when \(w_i=1\). Vertex names and bits are shown
for illustration only; the word is encoded entirely by the
structure of the unlabeled graph.}
\label{fig:word-encoding}
\end{figure}

The vertex \(r\) has degree \(n\), whereas every other
vertex has degree at most four. Hence \(r\) is uniquely determined by
the graph, and its neighborhood induces the original path. It follows
that
\[
    G_w\cong G_{w'}
    \quad\Longleftrightarrow\quad
    w'=w\ \text{or}\ w'=w^{\mathrm{rev}},
\]
where \(w^{\mathrm{rev}}\) denotes the reversal of \(w\).

For a language \(D\subseteq\bigcup_{n\geq6}\{0,1\}^n\) closed under
reversal, define the graph language
\[
    L_D=\{G:G\cong G_w\text{ for some }w\in D\}.
\]

\begin{lemma}\label{lem:ports-word-languages}
If \(D\) is decidable and closed under reversal, then
\(L_D\in\hier{\portsh}\).
\end{lemma}

\begin{proof}
Fix an arbitrary port assignment. The prover first marks two
vertices \(r,s\), two sets \(P,F\subseteq V(G)\), and chooses
a bit labeling \(b: P \to \{0,1\}\). Each vertex marking is represented by
one bit per vertex. We require that exactly one vertex receives
each of the marks for \(r\) and \(s\), and impose the following
conditions:
\begin{enumerate}
    \item \(\{r\},P,F\) partition \(V(G)\), and \(N(r)=P\).
    \item \(G[P]\) is a path on at least six vertices,
    and \(s\) is one of its endpoints.
    \item Every vertex in \(F\) has degree one and its neighbor
    belongs to \(P\); every vertex in \(P\) has at most one
    neighbor in \(F\).
\end{enumerate}
These conditions are MSO-definable and characterize the graphs
\(G_w\), together with a choice of an endpoint \(s\).
Write the path order starting at \(s\) as \(v_1,\ldots,v_n\),
and let \(w_i=1\) if and only if \(v_i\) has a neighbor in \(F\). Let \(p_i\) be the neighbor of \(r\) at its local port \(i\).
The order \(p_1,\ldots,p_n\) may differ from the path order
\(v_1,\ldots,v_n\). The labeling \(b\) defines the word
\[
    z=b(p_1)\cdots b(p_n).
\]
Vertex \(r\) can read \(z\) from its \portsh view and decide
whether \(z\in D\). It remains to verify that \(z=w\).

The disprover chooses a vertex set \(U\), and the prover responds
with a vertex set \(T\) and two bit labelings \(a,h\).
In addition to the conditions on the initial markings, we require:
\begin{enumerate}
    \item If \(U=\{u\}\) for some \(u\in P\), then \(h\equiv1\)
    and \(T\) is exactly the vertex set of the subpath from
    \(s\) to \(u\) in \(G[P]\), including both endpoints.
    \item Otherwise, \(h\equiv0\).
\end{enumerate}
Here \(h\equiv\varepsilon\) means that every vertex receives
the bit \(\varepsilon\). In the second case, there is no
restriction on \(T\).

Let \(\varphi\) be the conjunction of the conditions on the
initial markings, including the uniqueness of the marks for
\(r\) and \(s\), and these requirements on \(U,T,h\).
This is an MSO condition. In particular, the condition on \(T\)
can be expressed as follows: if \(s=u\), require \(T=\{s\}\);
otherwise, require that \(T\subseteq P\) induces a connected
path whose endpoints are exactly \(s\) and \(u\).

In addition to verifying \(\varphi\), the final verifier performs
the following local checks:
\begin{enumerate}
    \item Every vertex checks that its bit \(a\) agrees with
    the bits \(a\) of all its neighbors.
    \item If \(h(v)=1\) and \(v\in U\), vertex \(v\) checks
    that \(a(v)=1\) if and only if it has a neighbor in \(F\).
    \item Vertex \(r\) always checks that \(z\in D\).
    If \(h(r)=1\), it also computes
    \[
        i=|N(r)\cap T|
    \]
    and checks that \(1\leq i\leq n\) and \(b(p_i)=a(r)\).
\end{enumerate}
Every vertex carrying the mark for \(r\) performs the third
check; the condition \(\varphi\) ensures that there is exactly
one such vertex. These checks use only a vertex's own markings
and the markings of its neighbors indexed by its local ports.
Since the graph is connected, the first check forces \(a\)
to be constant.

The protocol therefore uses the prefix \(\exists r,s,P,F,b, \forall U, \exists T,a,h,\) where quantification over \(r,s\) abbreviates the choice of
the binary markings described above. This prefix is followed
by a \seth protocol for \(\varphi\), which exists by
\Cref{thm:set-equals-mso}. The markings chosen before this
additional protocol are treated as fixed finite input labels.
The final \portsh verifier conjoins the outcome of that
protocol with the local checks above. All earlier markings
remain available in the certificate sequences. The resulting
protocol uses a constant number of rounds and a constant
number of certificate bits per vertex, including those
needed to verify \(\varphi\).

For completeness, suppose that \(G=G_w\) with \(w\in D\).
Mark \(r,P,F\) according to the construction of \(G_w\),
and choose an endpoint \(s\) of the path. Let
\(v_1,\ldots,v_n\) be the order starting at \(s\), and write
\(w\) for the word read in this order. Since \(D\) is closed
under reversal, this word belongs to \(D\). Choose \(b\) so that \( b(p_i)=w_i\) for every  \(i\in[n]\). Thus \(z=w\), and vertex \(r\) accepts its membership check.

If \(U\) is not a singleton contained in \(P\), choose
\(h\equiv0\), \(a\equiv0\), and \(T=\varnothing\).
Then \(\varphi\) and all local checks hold.
If \(U=\{v_j\}\), choose
\(h\equiv1\), \(T=\{v_1,\ldots,v_j\}\), \( a\equiv w_j\). Again \(\varphi\) holds. The check at \(v_j\) succeeds by the
definition of \(w_j\), and the check at \(r\) succeeds because \( |N(r)\cap T|=j\) and \(
    b(p_j)=w_j=a(r).\)
In both cases, the prover can win the appended \seth protocol
because \(\varphi\) is true. This proves completeness.

For soundness, suppose that the prover has a winning strategy.
The appended \seth protocol forces the conditions on the
initial markings. Thus these markings determine a graph
\(G_w\) and a path order \(v_1,\ldots,v_n\).
The initial labeling \(b\) fixes the word \( z=b(p_1)\cdots b(p_n),\)
and the unconditional check at \(r\) gives \(z\in D\).

For each \(j\in[n]\), consider the choice \(U=\{v_j\}\).
The condition \(\varphi\) forces
\( h\equiv1 \)   and \(
    T=\{v_1,\ldots,v_j\}.\)
The local check at \(v_j\) gives \(a(v_j)=w_j\), and the
checks along the edges force \(a\) to be constant. Hence
\(a(r)=w_j\). Since \(|N(r)\cap T|=j\), the check at \(r\)
then gives
\[
    z_j=b(p_j)=a(r)=w_j.
\]
The labeling \(b\) was chosen before \(U\), so the same word
\(z\) satisfies this equality for every \(j\). Thus \(z=w\).
Since \(z\in D\), we conclude that \(G\in L_D\).

The port assignment was arbitrary. The initial labeling \(b\)
may depend on it, and the argument works for every assignment.
The decider for \(D\) is used only in the local computation
at \(r\), whose running time is unrestricted in our model.
\end{proof}

\begin{lemma}\label{lem:multiset-word-diagonal}
There exists a decidable language
\(D\subseteq\bigcup_{n\geq6}\{0,1\}^n\), closed under reversal, such
that \(L_D\notin\hier{\multiseth}\).
\end{lemma}

\begin{proof}
Every \multiseth protocol can be represented using a finite sequence
of quantified binary labelings
\[
    Q_1\beta_1\cdots Q_c\beta_c,
    \qquad Q_i\in\{\exists,\forall\},
    \qquad\beta_i:V\to\{0,1\},
\]
where we do not require consecutive quantifiers to alternate. Indeed,
for each original certificate alphabet \(\Gamma\), choose a
surjection \(\{0,1\}^d\to\Gamma\), replace its quantified labeling
by \(d\) consecutive binary labelings with the same quantifier, and
let the verifier decode each complete certificate sequence. The
induced map on global labelings is surjective, so this replacement
preserves both existential and universal quantification. An ignored
binary labeling allows us to assume \(c\geq1\).

Fix \(c\), and put \(q=2^c\). On a graph of maximum degree at most
\(n\), the information available to a \multiseth verifier at a vertex
is determined by its own certificate sequence and the multiplicity
of each of the \(q\) possible neighboring sequences. Since every
multiplicity lies in \(\{0,\ldots,n\}\), the number of possible
local views is at most
\(p_c(n)=q(n+1)^q.\)
There are no input labels. In particular, this bound applies to
every graph \(G_w\) with \(|w|=n\geq6\).

Let \(\mathcal R_n\) contain the lexicographically smaller word in
each reversal class of \(\{0,1\}^n\), listed in lexicographic order,
and let \(r_n=|\mathcal R_n|\). Since each class has at most two
elements, \(r_n\geq2^{n-1}.\)
These representatives encode pairwise nonisomorphic graphs.

For a fixed quantifier sequence, the behavior of any verifier on
these graphs is determined by a Boolean table on at most
\(p_c(n)\) local views. Thus there are at most \(2^{p_c(n)}\)
possible tables. Including the \(2^c\) quantifier sequences, protocols
with \(c\) binary labelings induce at most
\(2^{p_c(n)+c}\)
acceptance vectors on \(\{G_w:w\in\mathcal R_n\}\). This bound
does not impose any restriction on the verifier's computation time.

Choose a strictly increasing computable sequence
\(n_1,n_2,\ldots\), with \(n_c\geq6\), such that
\[
    p_c(n_c)+c<2^{n_c-1}.
\]
This is possible because \(p_c\) is a polynomial for each fixed
\(c\): at stage \(c\), take the least integer larger than the
previous term, and at least six, satisfying the inequality.

At stage \(c\), enumerate every quantifier sequence of length \(c\)
and every Boolean table on the finite set of possible local views
of degree at most \(n_c\). Evaluate the resulting finite
quantification games on every \(G_w\) with
\(w\in\mathcal R_{n_c}\). This computation terminates, since all
labeling domains are finite and every local decision is supplied by
the enumerated table. There are at most \(2^{p_c(n_c)+c}\) resulting
acceptance vectors, whereas
\[
    2^{p_c(n_c)+c}<2^{2^{n_c-1}}\leq2^{r_{n_c}}.
\]
Let \(\eta_c\in\{0,1\}^{r_{n_c}}\) be the lexicographically first
vector not obtained in this enumeration.

Define \(D\) to reject all words whose length is not one of the
\(n_c\). At length \(n_c\), include a word \(w\) precisely when
the coordinate of \(\eta_c\) indexed by the representative of its
reversal class is one. This language is closed under reversal.
It is also decidable: given \(w\), generate the increasing sequence
\(n_c\) until its first term satisfying \(n_c\geq|w|\). If
\(n_c>|w|\), reject; otherwise, compute \(\eta_c\) by the finite
procedure above and inspect the coordinate corresponding to the
reversal class of \(w\). Every step of this procedure terminates.

Suppose now that a \multiseth protocol recognizes \(L_D\), and put
it into binary form with \(c\geq1\) labelings. Its restriction to
\(\{G_w:w\in\mathcal R_{n_c}\}\) is described by one of the
quantifier sequences and tables enumerated at stage \(c\).
Its acceptance vector therefore cannot be \(\eta_c\). But the
definition of \(D\) requires precisely this vector, a contradiction.
\end{proof}

\begin{proof}[Proof of \Cref{th:multiset-ports}]
Take the language \(D\) supplied by
\Cref{lem:multiset-word-diagonal}. By
\Cref{lem:ports-word-languages}, \(L_D\in\hier{\portsh}\), whereas
\(L_D\notin\hier{\multiseth}\). The graph language \(L_D\) is
decidable: recognize the encoding, extract either orientation of the
word, and run the decider for \(D\). Together with the immediate
inclusion \(\hier{\multiseth}\subseteq\hier{\portsh}\), this proves
the theorem.
\end{proof}

\section{Equivalence of hierarchies on bounded-degree graphs}
\label{sec:boundeddegree}
\label{sec:Hstrict}

We compare the expressiveness of our constant-size certification hierarchies
with the \emph{locally polynomial hierarchy} of
\cite{reiter2024local}. This comparison is useful specifically on
graphs of bounded degree; we therefore introduce the model  here, rather
than alongside the weak neighborhood views that are the subject of this paper.

In the model of \cite{reiter2024local}, a node may communicate for a fixed number of rounds
and perform a polynomial amount of local computation. Nodes have identifiers
that are unique within a fixed radius (not necessarily globally), and the
players alternately assign certificates whose sizes are polynomially bounded
in the information contained in constant-radius neighborhoods. Write
$\Sigma_m^{\mathrm{LP}}$ and $\Pi_m^{\mathrm{LP}}$ for the levels with $m$
alternating certificate assignments, beginning with an existential and a
universal assignment, respectively. The protocol is required to be correct
for every admissible identifier assignment. We use the formal definition
and the bounded-degree normalization results of~\cite{reiter2024local}.

For $\modelo\in\{\seth,\portsh\}$, we similarly write
$\Sigma_m^{\modelo}$ and $\Pi_m^{\modelo}$ for the two possible starting
quantifiers at level $m$ of our hierarchy; in particular,
$\hier{\modelo}_m=\Sigma_m^{\modelo}\cup\Pi_m^{\modelo}$ for $m\geq1$.
For a fixed $\Delta\geq1$, let $\mathcal B(\Delta)$ denote the class of
connected graphs of maximum degree at most $\Delta$. Throughout this section
we compare languages over a \emph{fixed finite input alphabet}; this
qualification is essential, since the certificate-size bound in the locally polynomial hierarchy
model can depend on the length of an unbounded input label. We use the notation
$\mathcal C|_{\mathcal B(\Delta)}$ for the restrictions of the languages in
$\mathcal C$ to $\mathcal B(\Delta)$.

\begin{lemma}\label{lem:ports-in-set-bounded}
Fix $\Delta\geq1$ and a finite input alphabet. For every $m\geq3$,
\[
\Sigma_m^{\seth}|_{\mathcal B(\Delta)}
=
\Sigma_m^{\portsh}|_{\mathcal B(\Delta)}
=
\Sigma_m^{\mathrm{LP}}|_{\mathcal B(\Delta)}.
\]
For every $m\geq4$,
\[
\Pi_m^{\seth}|_{\mathcal B(\Delta)}
=
\Pi_m^{\portsh}|_{\mathcal B(\Delta)}
=
\Pi_m^{\mathrm{LP}}|_{\mathcal B(\Delta)}.
\]
The assertion is uniform over each fixed choice of the communication radius
and the identifier-uniqueness radius of the locally polynomial verifier.
\end{lemma}

\begin{proof}
A \seth verifier is simulated by a \portsh verifier simply by discarding
port order. A constant-time locally polynomial machine can in turn simulate a
\portsh verifier: it orders the neighbors according to the available LOCAL
interface (or, equivalently, their locally distinct identifiers). This is
legitimate because the \portsh protocol must recognize the same language
for every local port numbering. Its certificates are constant-size and hence
satisfy the local polynomial bound in \cite{reiter2024local}. Thus only the inclusion from
$\mathrm{LP}$ into \seth requires proof.

\paragraph*{Coloring and locally certified views.}
Let $M$ be a locally polynomial verifier for a language $L$ on
$\mathcal B(\Delta)$. Fix its communication time, identifier-uniqueness
radius, certificate-bound radius, and finite input alphabet. We choose a
constant $R\geq1$ large enough that a correctly reconstructed radius-$R$
view contains everything needed to simulate $M$ and that uniqueness of
identifiers within distance $2R$ implies its identifier assumption.

An auxiliary labeling $\chi$ will serve as a synthetic identifier
assignment. We require it to be a proper coloring of $G^{2R}$, with colors
from the constant-size palette $[C_{\Delta,R}]$, where
\[
C_{\Delta,R}
=1+\Delta\sum_{j=0}^{2R-1}(\Delta-1)^j.
\]
Indeed, the number of vertices at distance at most $2R$ from any vertex is
at most $C_{\Delta,R}$; greedily coloring $G^{2R}$ gives the bound.
For example, if $\Delta=2$, this is $4R+1$. The palette is fixed
independently of the number of vertices.

We record how a radius-$R$ LOCAL view can be simulated with radius-one
\seth checks. Let $\kappa_1,\ldots,\kappa_m$ be the original certificate
assignments and let $\rho$ be an additional bit labeling when needed.
An existential certificate, issued \emph{after} these assignments, supplies
finite descriptions $U_j(v)$, for $0\leq j\leq R$, of the labeled
radius-$j$ unfoldings at each node. The base description contains the
node's own input, color $\chi(v)$, original certificates, and $\rho(v)$.
For definiteness, one may choose the recursively defined descriptions
\[
U_0(v)=(\chi(v),x(v),\kappa_1(v),\ldots,\kappa_m(v),\rho(v)),
\qquad
U_j(v)=\left(U_0(v),\{U_{j-1}(u):u\in N(v)\}\right).
\]
Each node checks the base description against its own labels and checks
every recurrence using the set of neighboring descriptions.

We now explain why these descriptions contain
enough information to reconstruct the local
topology required by $M$.

Fix a vertex $v$ and suppose that $\chi$
is a proper coloring of $G^{2R}$.
Since any two vertices in $B_R(v)$ are at
distance at most $2R$, distinct vertices
in this ball have distinct colors.

The recursive description $U_R(v)$
contains the labeled unfoldings along all
walks of length at most $R$ starting at $v$.
Consequently, every vertex in $B_R(v)$
appears in this description.

Moreover, two occurrences correspond to
the same vertex if and only if they carry
the same $\chi$-color. Thus the colors
allow us to identify repeated occurrences
of vertices reached through different walks.

For every vertex $u\in B_{R-1}(v)$,
the description also contains $U_1(u)$,
which records the colors and labels of all
neighbors of $u$. Therefore, we can
recover all adjacencies between vertices
of $B_{R-1}(v)$.

It follows that $U_R(v)$ determines the
labeled induced subgraph $G[B_{R-1}(v)]$,
with $\chi$ serving as an injective
identifier assignment within this ball.

Choosing $R$ strictly larger than the
communication radius of $M$, and large
enough to satisfy its identifier-uniqueness
requirements, ensures that the reconstructed
view contains all the information needed
to simulate $M$.

If a local ordering of neighbors is needed,
we use the order induced by their distinct
$\chi$-colors.

Because $\Delta$, $R$, and all label alphabets are fixed, these
unfoldings have bounded size. More importantly, if \emph{all} local
recurrence checks accept, induction on $j$ proves that every claimed
unfolding is correct. This is a global conjunction of local checks; we
do not assume that one vertex can independently inspect a distant
vertex's true label. Thus a final existential assignment can provide
and validate the necessary higher-radius information without adding
a quantifier round.

For later use, a universal challenge $\rho:V\to\{0,1\}$ can test $\chi$.
At a node $v$, inspect the certified radius-$R$ unfolding and reject
if it contains two vertex occurrences of the same $\chi$-color carrying
different $\rho$-bits. If $\chi$ is proper on $G^{2R}$, this never
happens: any two vertices in $B_R(v)$ are at distance at most $2R$.
Conversely, if distinct vertices $u,w$ at distance at most $2R$ have
the same color, choose a midpoint $v$ on a shortest $u$--$w$ path
so that $u,w\in B_R(v)$. Assign $\rho(u)=0$ and $\rho(z)=1$ for all
$z\ne u$. The correct unfolding at $v$ now reveals a collision.
An existentially supplied false unfolding cannot circumvent this
challenge, because some recurrence or base check would reject.

\paragraph*{Protocols beginning existentially.}
First let $m\geq3$ be odd, so the $m$-th assignment is existential.
In the first round the prover chooses $\chi$ in addition to the
first certificate $\kappa_1$ of $M$. In the second round the
adversary chooses $\rho$ together with $\kappa_2$. The remaining
rounds carry the respective original certificates. In the final
existential round, the prover also supplies the unfoldings $U_j$.
The one-hop verifier checks all unfolding recurrences, tests the
$\rho$-challenge on $\chi$, and simulates the local acceptance
predicate of $M$ with the reconstructed view.

If $\chi$ is valid, every challenge passes and the simulation
has exactly the truth value of the original locally polynomial
protocol for that admissible identifier assignment. If $\chi$ is
invalid, the second player can choose a challenge exposing a
collision; no later existential reply can make all vertices accept.
The quantifier pattern is unchanged, and all certificates are bounded.
Therefore
$\Sigma_m^{\mathrm{LP}}|_{\mathcal B(\Delta)}
\subseteq\Sigma_m^{\seth}|_{\mathcal B(\Delta)}$
for odd $m\geq3$.

For even $m\geq4$, we use the bounded-degree elimination of a
trailing universal certificate (Proposition~38 of
\cite{reiter2024local}):
$\Sigma_m^{\mathrm{LP}}|_{\mathcal B(\Delta)}
=\Sigma_{m-1}^{\mathrm{LP}}|_{\mathcal B(\Delta)}$.
Apply the preceding construction at the odd level $m-1\geq3$
and append a dummy universal assignment. This establishes the
claimed $\Sigma_m$ inclusion for all $m\geq3$.

\paragraph*{Protocols beginning universally.}
Suppose first that $m\geq4$ is even, so the last certificate
assignment is existential. The first (universal) assignment now
contains both $\chi$ and $\kappa_1$. The second (existential)
assignment contains $\kappa_2$, a bit $b$, and an auxiliary witness.
All nodes check that $b$ has the same value at both ends of every
edge; connectedness makes $b$ global.

The meaning of the bit is the following.
If $b=1$, the prover claims that $\chi$ is valid. The third
(universal) round adds the challenge $\rho$ to $\kappa_3$.
The verifier tests the coloring exactly as above and, using
unfoldings supplied in the last existential round, simulates $M$.
An invalid coloring is therefore refuted by a suitable $\rho$.

If $b=0$, the prover claims that $\chi$ is invalid. Its second-round
witness consists of a number $d\in[2R]$, a color $a$, and
node roles from $\{\bot,0,\ldots,d\}$ (the values $d,a$ are
required to be uniform). Every vertex with role $i<d$ checks that
it has a neighbor with role $i+1$; a vertex with role $0$ or $d$
checks that its $\chi$-color is $a$. Because each vertex has only
one role, the existence of a role-$0$ vertex forces a path with
$d$ distinct vertices after its start, whose two distinct endpoints
have the same $\chi$-color. Thus a locally valid witness with a
role-$0$ vertex certifies a violation of proper coloring of $G^{2R}$.
Conversely, a shortest path between two equally colored vertices
at distance $d\leq2R$ provides such a witness.

It remains to require \emph{at least one} role-$0$ vertex. This is the
language $\alos$, which has a three-round \seth protocol with
quantifier sequence $\exists\,\forall\,\exists$
(\Cref{th:alos}). Its rounds are merged with rounds two, three,
and four of the present protocol. The one-hop verifier enforces
the path constraints and the $\alos$ protocol in the $b=0$ branch;
when $b=0$ it ignores $M$. The auxiliary challenge and response can
be padded with dummy fields in any subsequent rounds.

If $\chi$ is valid, choosing $b=0$ cannot succeed: every candidate
path either fails a local constraint or has no role-$0$ vertex,
in which case the third-round adversary defeats the $\alos$ test.
The prover must choose $b=1$, where the correct simulation of $M$
determines acceptance. If $\chi$ is invalid, the prover can choose
$b=0$ and certify an actual color collision, so the invalid
identifier assignment does not affect membership. Therefore the
universal first-round choice ranges, effectively, only over valid
identifiers. In the $b=1$ branch, the last existential assignment
supplies and validates the necessary radius-$R$ unfoldings, exactly
as in the $\Sigma$ case. The original $m$ quantifier rounds suffice.

For odd $m\geq5$, Proposition~38 in \cite{reiter2024local}  gives
$\Pi_m^{\mathrm{LP}}|_{\mathcal B(\Delta)}
=\Pi_{m-1}^{\mathrm{LP}}|_{\mathcal B(\Delta)}$.
The even-level construction applies to $m-1\geq4$ and can be
padded by one dummy universal round. This completes the $\Pi_m$
inclusion and the proof of the lemma.
\end{proof}

The restriction $m\geq3$ in the $\Sigma$ statement is deliberate:
our radius-one reconstruction uses a final existential assignment,
which is not available in the two-round $\exists\forall$ pattern.
The separation argument below does not require an equivalence at
that level. Likewise, the lower levels of the $\Pi$ hierarchy are
not needed here. Notice also that the horizon $R$ is a fixed
constant depending on the simulated verifier, not a common radius
imposed on all languages.

\thmHstrict*

\begin{proof}

We prove the stronger statement that, for every \(k\geq0\),
there exists a language
\(L_k\in\hier{\seth}_{k+1}\setminus\hier{\portsh}_k\).
This implies the theorem, since
\(\hier{\seth}_j\subseteq\hier{\multiseth}_j
\subseteq\hier{\portsh}_j\)
for every \(j\geq0\).

We first consider \(k\geq2\).
Fix \(\mathcal B=\mathcal B(4)\), the class of connected
graphs of maximum degree at most four, and restrict all
the classes below to \(\mathcal B\).

By Theorem~36 of~\cite{reiter2024local}, for every even
integer \(\ell\geq2\),
\[
    \Sigma_{\ell-1}^{\mathrm{LP}}
    \subsetneq
    \Pi_\ell^{\mathrm{LP}}
    \subsetneq
    \Sigma_{\ell+1}^{\mathrm{LP}}.
\]
These separations also hold on our class of connected
graphs: the graph encodings of rectangular pictures
used to establish them are connected and have maximum
degree four.

By Proposition~38 of~\cite{reiter2024local},
the final universal quantifier can be eliminated on
bounded-degree graphs. Consequently, for even \(k\geq2\),
we have
\(\Sigma_k^{\mathrm{LP}}=\Sigma_{k-1}^{\mathrm{LP}}
\subseteq\Pi_k^{\mathrm{LP}}\),
while for odd \(k\geq3\),
\(\Pi_k^{\mathrm{LP}}=\Pi_{k-1}^{\mathrm{LP}}
\subseteq\Sigma_k^{\mathrm{LP}}\).
The inclusions follow by adding a dummy initial
quantifier.

Since a \portsh verifier can be simulated in the
locally polynomial model, we obtain
\(\hier{\portsh}_k\subseteq\Pi_k^{\mathrm{LP}}\)
when \(k\) is even, and
\(\hier{\portsh}_k\subseteq\Sigma_k^{\mathrm{LP}}\)
when \(k\) is odd.

Suppose first that \(k\geq2\) is even. By Theorem~36,
there exists a language
\(L\in\Sigma_{k+1}^{\mathrm{LP}}
\setminus\Pi_k^{\mathrm{LP}}\).
By \Cref{lem:ports-in-set-bounded}, the restriction of
\(L\) to \(\mathcal B\) coincides with the restriction
of some language \(L_k\in\Sigma_{k+1}^{\seth}\).

If \(L_k\) belonged to \(\hier{\portsh}_k\), then its
restriction to \(\mathcal B\) would belong to
\(\Pi_k^{\mathrm{LP}}\), a contradiction. Hence
\(L_k\in\hier{\seth}_{k+1}
\setminus\hier{\portsh}_k\).

Now suppose that \(k\geq3\) is odd.
Theorem~36 gives a language
\(L\in\Pi_{k+1}^{\mathrm{LP}}
\setminus\Sigma_k^{\mathrm{LP}}\).
Since \(k+1\geq4\), \Cref{lem:ports-in-set-bounded}
provides a language \(L_k\in\Pi_{k+1}^{\seth}\)
whose restriction to \(\mathcal B\) coincides
with that of \(L\). The same argument gives
\(L_k\in\hier{\seth}_{k+1}
\setminus\hier{\portsh}_k\).

It remains to consider \(k=0\) and \(k=1\).

For \(k=0\), consider the language of bipartite graphs.
It belongs to \(\hier{\seth}_1\), since the prover can
assign a proper \(2\)-coloring. On the other hand,
a \portsh verifier without certificates cannot
distinguish \(C_3\) from \(C_4\): in both graphs every
node has degree two and receives the same trivial
input. Thus
\(\textnormal{Bipartiteness}\in
\hier{\seth}_1\setminus\hier{\portsh}_0\).

For \(k=1\), consider \(\amos\).
By \Cref{th:amos}, we have
\(\amos\in\hier{\seth}_2\).
We show that \(\amos\notin\hier{\portsh}_1\).

First, \(\amos\notin\Sigma_1^{\portsh}\).
Consider a cycle \(C_n\), with \(n\geq4\), containing
exactly one marked vertex. Suppose that an existential
one-round \portsh protocol accepts this instance.
Fix a port numbering and an accepting certificate
assignment.

Take the double cover \(C_{2n}\) of \(C_n\), lifting
the input marks, port numbering, and certificates.
The covering preserves every local \portsh view,
so every vertex of \(C_{2n}\) accepts. However,
the lifted instance has two marked vertices and
does not belong to \(\amos\), a contradiction.

Second, observe that
\(\Pi_1^{\portsh}=\hier{\portsh}_0\).
Indeed, a universally quantified certificate assignment
can be eliminated locally. Each vertex enumerates all
possible certificates for itself and its neighbors,
indexed by its local ports, and accepts if the original
verifier accepts for every such assignment.
Every local assignment extends to a global one, so
this test preserves the universal acceptance condition.
Since local computation is unrestricted, the test can
be performed without certificates.

Finally, \(\amos\notin\hier{\portsh}_0\).
On cycles, a verifier without certificates sees only
its own input and its two ports. An accepting cycle
with one marked vertex therefore forces acceptance
of both possible local input values, and the same
verifier accepts a cycle with two marked vertices.
Consequently,
\(\amos\in\hier{\seth}_2\setminus\hier{\portsh}_1\).

We have thus established, for every \(k\geq0\),
the existence of a language
\(L_k\in\hier{\seth}_{k+1}\setminus\hier{\portsh}_k\).
By the level-wise inclusions between our three
models, this proves that every consecutive inclusion
is strict in each hierarchy.
\end{proof}

\section{Explicit languages outside the \(\portsh\) hierarchy}
\label{sec:sizeout}

In this section, we give explicit examples of languages outside the
\portsh hierarchy. We identify two limitations of constant-size
alternating certification in this model. The first arises on highly
symmetric graphs of unbounded degree, even with a finite input alphabet;
the second arises from unbounded local inputs, even on bounded-degree
graphs.

\subsection{Labeled complete bipartite graphs}
\label{sec:knnout}

\newcommand\compbknn[1]{\mathcal{K}_{#1}}
Let \(L\) be a distributed language over a fixed finite input alphabet
\(Q\). We associate with \(L\) the following two-party communication
problem, defined for every \(n,m\geq 1\):
\[
\compbknn{L}=\bigl\{(u,v)\in Q^n\times Q^m:
(K_{n,m},x_{u,v})\in L\bigr\}.
\]
Here Alice receives \(u=(u_1,\ldots,u_n)\), the inputs of the \(n\)
vertices in one bipartition class, and Bob receives
\(v=(v_1,\ldots,v_m)\), the inputs of the \(m\) vertices in the other.
Each player knows the size of their own input but not necessarily that
of the other player. Communication complexity is measured uniformly
over all \(n,m\).

We show that an arbitrary fixed number of alternating rounds with
constant-size certificates in the \portsh model yields a uniformly
bounded deterministic communication protocol on these graphs. We use
the fact that nondeterministic communication protocols of constant
cost can be replaced by deterministic protocols of constant cost
(the equality \(P_0=NP_0\); see~\cite{HatamiHatami2024}).

\begin{lemma}\label{lem:boundedcc}
Let \(L\) be a distributed language over a fixed finite input alphabet
\(Q\). If \(L\in\hier{\portsh}\), then \(\compbknn{L}\) has uniformly
bounded deterministic communication complexity. That is, there exists
a constant \(B\), independent of \(n,m\), such that Alice and Bob can
decide membership in \(\compbknn{L}\) by exchanging at most \(B\) bits.
\end{lemma}

\begin{proof}
Fix a \portsh protocol for \(L\), with a finite number of alternating
certificate assignments, each over a finite alphabet. Membership is
independent of the choice of local port numbering, so on \(K_{n,m}\)
we may fix a convenient one: at every vertex of Alice's class, port
\(j\) leads to Bob's vertex \(j\), and at every vertex of Bob's class,
port \(i\) leads to Alice's vertex \(i\).

We argue by induction on the number of certificate quantifiers that
remain to be evaluated. For precision, we allow certificate fields
whose quantifiers have already been removed to be supplied as
additional, finite-valued distributed inputs to the communication
problem, while retaining their original visibility at neighboring
ports. This strengthened induction statement is useful because a
certificate visible to neighboring vertices cannot simply be treated
as an ordinary local input in the definition of \portsh. All the
additional alphabets remain finite. Equivalently, such information
could be copied into visible certificates and checked against the
corresponding local inputs with an existential certificate assignment.

For the base case, all certificate fields are fixed. Each player knows
the inputs and the certificate fields of the vertices in their own
class. Let \(T_A\) be the set of local types occurring in Alice's
class, where a type consists of a vertex input and all its own
certificate fields; define \(T_B\) analogously. Both sets range over
a fixed finite universe. Alice sends \(T_A\) to Bob, and Bob sends
\(T_B\) to Alice, using a constant number of bits.

Because of the chosen port numbering, all vertices of Alice's class
see the same port-ordered tuple of certificate fields from Bob's
class. Bob knows that tuple, as well as its length \(m\). He can
therefore evaluate the local verifier on every type in \(T_A\) and
determine whether all vertices of Alice's class accept. Symmetrically,
Alice can determine whether all vertices of Bob's class accept. They
exchange the resulting two bits and decide whether every vertex
accepts. This is a deterministic protocol with uniformly bounded
communication. Notice that the verifier does not need to read the
original inputs of neighboring vertices: it reads only their visible
certificate fields, exactly as specified by \portsh.

For the induction step, first suppose that the next certificate
assignment is existential. Write the residual condition as
\[
R(u,v)\iff \exists y\,R'((u,y_A),(v,y_B)),
\]
where \(y_A\) and \(y_B\) are the restrictions of the finite-alphabet
labeling \(y\) to Alice's and Bob's vertices, respectively. The
induction hypothesis provides a deterministic communication protocol
\(P'\), of uniformly bounded cost, for \(R'\), with the new
certificate field still visible at neighboring ports.

Alice and Bob now use the following nondeterministic protocol for
\(R\). First, they guess an accepting communication transcript
\(\mu\) of \(P'\). There are only finitely many possible transcripts,
since \(P'\) communicates at most a constant number of bits. Alice
checks, using only her own input and \(\mu\), whether there exists
some assignment \(y_A\) for which all messages she would send in
\(P'\) are consistent with \(\mu\). Bob independently performs the
analogous check for some \(y_B\). They accept if both checks succeed.

If they accept, the two local assignments together define a global
labeling \(y\), and the complete transcript \(\mu\) is consistent
with the deterministic protocol \(P'\) on that input. Since the
transcript is accepting, \(R'\) holds, hence \(R\) holds. Conversely,
if \(R\) holds, take a witnessing labeling \(y\) and the accepting
transcript of \(P'\) on that labeling. Both local consistency checks
then succeed. Thus this nondeterministic protocol decides \(R\) with
constant cost. By \(P_0=NP_0\), \(R\) also has a deterministic
constant-cost protocol.

If the next certificate assignment is universal, use
\[
\forall y\,R'(u,v,y)
\quad\Longleftrightarrow\quad
\neg\exists y\,\neg R'(u,v,y).
\]
By induction, \(R'\) has a bounded deterministic protocol;
complementation preserves its cost. The existential case therefore
applies to \(\neg R'\), and a final complementation gives a bounded
deterministic protocol for the universal case as well.

Applying the induction through the entire fixed quantifier prefix
proves the lemma.
\end{proof}

\begin{corollary}\label{cor:bipartite-outside-ports}
Let \(L\) be a graph language such that, for some constant \(n_0\geq1\),
\[
K_{n,m}\in L\quad\Longleftrightarrow\quad n=m
\qquad\text{for all }n,m\geq n_0.
\]
Then \(L\notin\hier{\portsh}\). In particular, neither the language
of Hamiltonian graphs nor the language of graphs admitting a perfect
matching belongs to the \portsh hierarchy.
\end{corollary}

\begin{proof}
Suppose \(L\in\hier{\portsh}\). Since \(L\) is a graph language,
its input alphabet can be taken to be a singleton. By
\Cref{lem:boundedcc}, Alice and Bob could decide whether \(n=m\)
for all \(n,m\geq n_0\) using a uniformly bounded number of bits.

This is impossible. Indeed, a deterministic protocol exchanging at
most \(B\) bits has only finitely many accepting transcripts, and
each accepting transcript defines a combinatorial rectangle. A
rectangle contained in the equality relation cannot contain two
distinct diagonal pairs \((n,n)\) and \((m,m)\), since it would then
also contain the non-diagonal pair \((n,m)\). Thus equality on
arbitrarily large positive integers requires infinitely many
accepting rectangles and cannot have uniformly bounded deterministic
communication complexity.

Finally, \(K_{n,m}\) admits a perfect matching exactly when \(n=m\).
It admits a Hamiltonian cycle exactly when \(n=m\geq2\). The
corollary follows by taking \(n_0=1\) and \(n_0=2\), respectively.
\end{proof}

This obstruction is specific to the information available in our
constant-certificate port model. It does not automatically extend to
the locally polynomial hierarchy of~\cite{reiter2024local}, which
allows locally unique identifiers and more general certificate-size
bounds.

\subsection{Bounded-degree graphs with unbounded inputs}
\label{sec:fragile-inputs}

We now consider a different obstruction. Local verifiers can process
arbitrarily long inputs at their own vertices, but in \portsh they
cannot read neighboring inputs directly. For a fixed degree, the
information provided by constant-size certificates and their
port-ordered neighboring tuples ranges over a finite set.

Let \(L\) be a distributed language with inputs in \(\{0,1\}^*\).
Given an instance \(I=(G,x)\), a vertex \(v\in V(G)\) with input
\(s=x(v)\), and any \(s'\neq s\), define the \emph{perturbation}
\(\widetilde I=(G,\widetilde x_{v,s'})\) by
\[
\widetilde x_{v,s'}(u)=
\begin{cases}
x(u) & \text{if }u\neq v,\\
s' & \text{if }u=v.
\end{cases}
\]

For a fixed degree \(d>0\), we say that \(s\in\{0,1\}^*\) is a
\emph{fragile degree-\(d\) input value} of \(L\) if there exists a
yes-instance \(I=(G,x)\in L\) with a degree-\(d\) vertex \(v\)
satisfying \(x(v)=s\), such that every perturbation of \(I\) at \(v\)
produces a no-instance. We say that \(I\) \emph{witnesses} the
fragility of \(s\). Let \(\frag{d}{L}\) be the number, possibly
infinite, of fragile degree-\(d\) input values of \(L\).

\lemlangfrag*

\begin{proof}
Let \(L\in\hier{\portsh}\), and fix \(d>0\). There are a finite
number \(k\) of certificate quantifiers and a local verifier \(A\)
such that
\[
(G,x)\in L\iff
\quantif_1\ell_1\cdots\quantif_k\ell_k\;
\forall v\in V(G),\quad
A\bigl(x(v),\overline\ell(v),
\view{\portsh}{v}\bigr)=1,
\]
where \(\overline\ell=(\ell_1,\ldots,\ell_k)\) and all certificate
alphabets are finite.

Let \(C\) be the finite alphabet of complete certificate tuples
\(\overline\ell(v)\). At a degree-\(d\) vertex, the certificate
information seen by \(A\) belongs to the finite set
\[
\mathcal V_d=C\times C^d,
\]
consisting of its own certificate tuple and the ordered tuples of
neighboring certificates indexed by its local ports. In particular,
for each possible input value \(s\), the verifier determines a
Boolean function
\[
F_s:\mathcal V_d\longrightarrow\{0,1\}.
\]
Only finitely many such functions exist. If
\(\frag{d}{L}=\infty\), we can therefore choose distinct fragile
values \(s\neq s'\) for which \(F_s=F_{s'}\).

Let \(I=(G,x)\in L\) witness the fragility of \(s\) at a degree-\(d\)
vertex \(v\), and replace only \(x(v)=s\) by \(s'\). By fragility,
the resulting instance \(\widetilde I\) is not in \(L\). However,
for every possible assignment of all certificates, the local output
at \(v\) is unchanged because \(F_s=F_{s'}\), and the local output
at every other vertex is unchanged because those vertices do not
read \(x(v)\) directly. Hence the statement that all vertices accept
has exactly the same truth value on \(I\) and \(\widetilde I\) for
every certificate assignment. The entire alternating quantified
statement therefore has the same truth value on both instances,
contradicting \(I\in L\) and \(\widetilde I\notin L\).

Thus \(\frag{d}{L}\) is finite.
\end{proof}

We now consider the distributed language
\[
\size=\{(G=(V,E),x):\forall v\in V,\ x(v)=|V|\},
\]
where integers are represented by their binary encodings. This
language asks every vertex to hold the correct number of vertices of
the graph.

\thmSize*

\begin{proof}
We show that \(\frag{2}{\size}=\infty\). For each \(n\geq3\),
consider the cycle \(C_n\) with input \(n\) at every vertex. This
is a yes-instance of \(\size\). Every vertex has degree \(2\), and
changing the input at any one vertex to a different value makes the
instance a no-instance. Consequently, every \(n\geq3\), encoded in
binary, is a fragile degree-\(2\) input value. Hence
\(\frag{2}{\size}=\infty\), and \Cref{lem:langfrag} implies
\(\size\notin\hier{\portsh}\).
\end{proof}

\begin{remark}
The same argument applies to other size-dependent inputs. For
example, let \(f:\mathbb N\to\mathbb N\) be an unbounded computable
function, and define
\[
L_f=\bigl\{(G=(V,E),x):\forall v\in V,\ x(v)=f(|V|)\bigr\}.
\]
The values \(f(n)\), for \(n\geq3\), are infinite in number, and the
cycle argument shows that each distinct value is a fragile
degree-\(2\) input. Thus \(L_f\notin\hier{\portsh}\).

For languages over a fixed finite input alphabet, this particular
fragility argument does not apply. Indeed, several finite-alphabet
properties, including agreement on a fixed input value and counting
modulo a fixed constant, can be verified within the hierarchy.
\end{remark}

\section{Discussion}
\label{sec:discussion}

Our results leave several natural questions open.
While the \seth hierarchy is exactly characterized
by MSO, a logical characterization of the
\multiseth and \portsh hierarchies remains unknown.
We have established the strict inclusions
\[
\hier{\seth}
\subsetneq
\hier{\multiseth}
\subsetneq
\hier{\portsh},
\]
showing that both multiplicities and port numbering
strictly increase expressive power. Moreover,
the \portsh hierarchy strictly contains CMSO,
while the \multiseth hierarchy can express
properties that are not CMSO-definable.
Whether every CMSO-definable language belongs
to the \multiseth hierarchy remains open.
More generally, finding logical characterizations
of these two hierarchies is an important direction
for future research.

Our negative results provide some insight into their
limitations. In particular, we have shown that
Hamiltonicity and the existence of a perfect matching
cannot be recognized in the \portsh hierarchy, and
therefore neither in the weaker \multiseth hierarchy.
Since Hamiltonicity is expressible in $\mathrm{MSO}_2$,
this establishes that even port numbering and arbitrarily
many alternating rounds of constant-size certificates
do not suffice to capture all properties definable using
edge-set quantification. This raises the question of
which fragments of $\mathrm{MSO}_2$, beyond CMSO, can be
captured by these hierarchies. Our result concerning
\size\ provides a different limitation, showing that
unbounded input alphabets can obstruct verification
even on bounded-degree graphs.

A second direction concerns the internal structure of
the hierarchies. Although we establish that the
hierarchies have infinitely many distinct levels,
explicit separations between small levels remain open.
In particular, it would be interesting to determine the
exact level of the language $\leader$. We show that
$\leader$ belongs to the third level, and it is known
that it does not belong to the first
level~\cite{goos2011locally}. Proving that it does not
belong to the second level would provide a natural and
explicit separation between the second and third levels.

Another natural direction is to understand what happens
between constant-size and logarithmic-size certificates.
In the hierarchy with logarithmic-size certificates,
unique identifiers can be assigned and certified at
some level of the hierarchy. Once identifiers are
available, nodes can distinguish their neighbors, and
multiplicities can be recovered from sets of decorated
local types. Consequently, the distinction between the
\seth and \multiseth views essentially disappears in
that setting.

This raises the question of what happens for intermediate
certificate sizes, for instance in the regime
$o(\log n)$. Is there a threshold below which
multiplicities still provide additional expressive
power? Or can one already recover enough symmetry-breaking
information to collapse the different hierarchies?
Understanding this intermediate regime could help
clarify which aspects of their separation are genuinely
due to constant-size certificates, and which are
consequences of anonymity itself.

Finally, larger certificates allow us to approach
stronger logical formalisms. Certificates of size
$\mathcal{O}(n\log n)$ are sufficient to encode arbitrary
edge-set assignments and capture $\mathrm{MSO}_2$.
On the other hand, logarithmic-size certificates already
suffice to recognize Hamiltonicity
~\cite{feuilloley2021hierarchy}. It remains open whether
such certificates are sufficient to capture every
$\mathrm{MSO}_2$-definable graph language, or whether
some properties require larger certificates

\bibliography{biblio}

\end{document}